\documentclass[aps,pra,reprint,superscriptaddress,10pt]{revtex4-1}
\usepackage[ascii]{inputenc}
\usepackage[bookmarksnumbered,hypertexnames=false,colorlinks=true,linkcolor=blue,urlcolor=blue,citecolor=blue,anchorcolor=green,breaklinks=true]{hyperref}
\usepackage{mathtools,amsmath,amsfonts,amssymb,amsbsy,amsthm}
\usepackage{graphicx}
\usepackage{multirow}
\usepackage{setspace}
\usepackage{tikz}
\usepackage[caption=false]{subfig}
\usepackage{enumitem}
\usepackage{braket}
\usepackage{lipsum}

\usepackage[capitalize]{cleveref}

\crefname{assump}{Assumption}{Assumptions}
\crefname{hypothesis}{Hypothesis}{Hypotheses}

\newtheoremstyle{plain}
{10pt}   
{10pt}   
{\itshape}  
{0pt}       
{\bfseries} 
{.}         
{5pt plus 1pt minus 1pt} 
{}
\newtheoremstyle{definition}
{10pt}   
{10pt}   
{}  
{0pt}       
{\bfseries} 
{.}         
{5pt plus 1pt minus 1pt} 
{}
\makeatletter
\renewenvironment{proof}[1][\proofname]{\par
  \vspace{-5pt}
  \pushQED{\qed}%
  \normalfont
  \topsep0pt \partopsep0pt 
  \trivlist
  \item[\hskip\labelsep
  \itshape
  #1\@addpunct{.}]\ignorespaces
}{%
  \popQED\endtrivlist\@endpefalse
  \addvspace{5pt plus 0pt} 
}
\def\@endtheorem{\endtrivlist}
\makeatother

\theoremstyle{plain}
\newtheorem{thm}{Theorem}
\crefname{thm}{Theorem}{Theorems}
\Crefname{thm}{Theorem}{Theorems}
\newtheorem{lemma}{Lemma}
\crefname{lemma}{Lemma}{Lemmas}

\crefname{proposition}{Proposition}{Propositions}
\newtheorem{dummy}{Dummy}
\newtheorem{corollary}{Corollary}[dummy]
\makeatletter
\newtheorem*{rep@theorem}{\rep@title}
\newcommand{\newreptheorem}[2]{%
  \newenvironment{rep#1}[1]{%
    \def\rep@title{#2 \ref{##1}}%
    \begin{rep@theorem}}%
    {\end{rep@theorem}}}
\makeatother
\newenvironment{theorem}{\thm}{\endthm}
\newreptheorem{theorem}{Theorem}
\theoremstyle{definition}

\theoremstyle{definition}

\usepackage{tabularx}
\newcolumntype{L}[1]{>{\hsize=#1\hsize\raggedright\arraybackslash}X}%
\newcolumntype{R}[1]{>{\hsize=#1\hsize\raggedleft\arraybackslash}X}%
\newcolumntype{C}[1]{>{\hsize=#1\hsize\centering\arraybackslash}X}%
\usepackage{booktabs}
\newcommand{\ra}[1]{\renewcommand{\arraystretch}{#1}}  
\AtBeginDocument{
	\heavyrulewidth=.1em
	\lightrulewidth=.05em
	\cmidrulewidth=.03em
	\belowrulesep=.65ex
	\belowbottomsep=0pt
	\aboverulesep=.4ex
	\abovetopsep=0pt
	\cmidrulesep=\doublerulesep
	\cmidrulekern=.5em
	\defaultaddspace=.5em
}

\newcommand{\ketbra}[1]{\ket{#1} \!\! \bra{#1}}

\newcommand{\Ic}{\ensuremath{\mathcal{I}}}
\newcommand{\Jc}{\ensuremath{\mathcal{J}}}
\newcommand{\Hc}{\ensuremath{\mathcal{H}}}
\newcommand{\Sc}{\ensuremath{\mathcal{S}}}

\newcommand{\Ebb}{\ensuremath{\mathbb{E}}}

\newcommand{\N}{\ensuremath{\mathbb{N}}}
\newcommand{\R}{\ensuremath{\mathbb{R}}}

\newcommand{\ot}{\otimes}
\newcommand{\smax}{\ensuremath{s_{\operatorname{max}}}}
\newcommand{\smin}{\ensuremath{s_{\operatorname{min}}}}
\newcommand{\ds}{\ensuremath{\Delta s}}
\newcommand{\half}{\ensuremath{\frac12}}
\DeclareMathOperator{\Tr}{Tr}
\DeclareMathOperator{\Binom}{Binom}
\allowdisplaybreaks[4]

\begin{document}
\title{Witnessing Entanglement in Experiments with Correlated Noise}
\date{\today}
\author{Bas Dirkse}
\affiliation{QuTech, Delft University of Technology, The Netherlands}
\affiliation{Kavli Institute of Nanoscience, Delft University of Technology, The Netherlands}
\affiliation{QuSoft, University of Amsterdam, The Netherlands}
\author{Matteo Pompili}
\affiliation{QuTech, Delft University of Technology, The Netherlands}
\affiliation{Kavli Institute of Nanoscience, Delft University of Technology, The Netherlands}
\author{Ronald Hanson}
\affiliation{QuTech, Delft University of Technology, The Netherlands}
\affiliation{Kavli Institute of Nanoscience, Delft University of Technology, The Netherlands}
\author{Michael Walter}
\affiliation{QuSoft, University of Amsterdam, The Netherlands}
\affiliation{Korteweg-de Vries Institute for Mathematics, Institute for Theoretical Physics, Institute for Logic, Language, and Computation, University of Amsterdam, The Netherlands}
\author{Stephanie Wehner}
\affiliation{QuTech, Delft University of Technology, The Netherlands}
\affiliation{Kavli Institute of Nanoscience, Delft University of Technology, The Netherlands}
\makeatletter
\hypersetup{pdftitle={\@title},pdfauthor={Bas Dirkse, Matteo Pompili, Ronald Hanson, Michael Walter, Stephanie Wehner}}
\makeatother

\begin{abstract}
The purpose of an entanglement witness experiment is to certify the creation of an entangled state from a finite number of trials.
The statistical confidence of such an experiment is typically expressed as the number of observed standard deviations of witness violations.
This method implicitly assumes that the noise is well-behaved so that the central limit theorem applies.
In this work, we propose two methods to analyze witness experiments where the states can be subject to arbitrarily correlated noise.
Our first method is a \emph{rejection experiment}, in which we certify the creation of entanglement by rejecting the hypothesis that the experiment can only produce separable states.
We quantify the statistical confidence by a~$p$-value, which can be interpreted as the likelihood that the observed data is consistent with the hypothesis that only separable states can be produced.
Hence a small~$p$-value implies large confidence in the witnessed entanglement.
The method applies to general witness experiments and can also be used to witness genuine multipartite entanglement.
Our second method is an \emph{estimation experiment}, in which we estimate and construct confidence intervals for the average witness value.
This confidence interval is statistically rigorous in the presence of correlated noise.
The method applies to general estimation problems, including fidelity estimation.
To account for systematic measurement and random setting generation errors, our model takes into account device imperfections and we show how this affects both methods of statistical analysis.
Finally, we illustrate the use of our methods with detailed examples based on a simulation of NV centers.
\end{abstract}

\maketitle

\section{Introduction}\label{sec:intro}
Entanglement is a fundamental property of quantum mechanical systems~\cite{Horodecki2009} and an important resource for many quantum information processing tasks.
In quantum computing, coherently creating entanglement between several qubits is necessary for computational speedups~\cite{Vidal2003,Nielsen2010,Horodecki2009}.
In quantum networks, remote entanglement is an essential resource for quantum cryptography~\cite{Shenoy-Hejamadi2017,Broadbent2016,Pirandola2019} and distributed computing applications~\cite{Ben-Or2005,Tani2012,Jozsa2000}.
Entanglement also plays a crucial role in quantum sensing and metrology~\cite{Bollinger1996,Gottesman2012,komar2014}, enabling more precise measurement of physical quantities.
With the rapid experimental advances in the manipulation and control of quantum systems, much progress had been made towards the generation of entangled states in various physical platforms~\cite{Bourennane2004,Filgueiras2012,Song2017,Friis2018,Zhong2018,Wang2018,Humphreys2018}.
Yet, the creation of high-quality many-body entanglement is still a significant challenge.
It is therefore important to have good tools to certify the creation of entanglement. The main tools used for this purpose are entanglement witnesses.

An \emph{entanglement witness} is an observable~$W$ on a quantum system that can certify entanglement of a state~$\rho^*$ under investigation~\cite{Guhne2009}.
By definition, the witness~$W$ satisfies
\begin{equation}\label{eq:witness_definition}
  \Tr[\rho W] \geq 0 \quad \forall \rho \in \Sc,
\end{equation}
for all separable states~$\rho \in \Sc$.
As a consequence, it can be used to certify entanglement:
If a state~$\rho^*$ has negative witness expectation value,~$\Tr[W \rho^*] < 0$, then it is necessarily entangled.
If the expectation value is non-negative, the test is inconclusive (the state can either be separable or not).
The witness method applies more generally than just for separating entangled from separable states.
If~$\Sc$ is an arbitrary convex set of states and~$\rho^*\not\in\Sc$, then there always exists a witness~$W$ such that \cref{eq:witness_definition} holds, while~$\Tr[W \rho^*] < 0$.
For example, a witness can be used to certify that states are genuinely multipartite entangled.
Geometrically, the witness~$W$ can be interpreted as a hyperplane that separates the convex set~$\Sc$ from the state~$\rho^*\notin\Sc$.
This is illustrated in \cref{fig:witnessschematicv2}.
In general, finding an appropriate witness~$W$ for a state~$\rho^*$ is a difficult problem that has been studied extensively in literature~\cite{Jungnitsch2011,Toth2009,Ryu2012}.
For the remainder of this article, we will assume that~$W$ is chosen and fixed.

\begin{figure}
	\centering
	\includegraphics[scale=1]{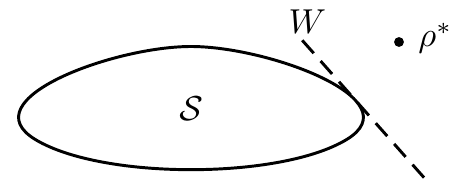}
	\caption{Geometric interpretation of a witness~$W$ as a hyperplane that separates the state~$\rho^*$ from the convex set~$\Sc$. If~$\Sc$ is the set of separable states, then~$W$ certifies that~$\rho^*$ is entangled.\label{fig:witnessschematicv2}}
\end{figure}

Often, a witness~$W$ is a non-local observable of the system for which entanglement is to be certified.
Such measurements are typically hard to perform, particularly in a network setting.
Therefore, in experiments,~$W$ is usually decomposed into a sum of locally measurable observables which are then measured individually on the constituent subsystems.
The \emph{witness expectation value}~$\Tr[W\rho^*]$ is then the sum of the expectation values of the locally measurable observables.
Each of these expectation values can then only be estimated to some finite precision, since in any experiment only a finite number~$n$ of data points can be collected.
As a consequence, the \emph{witness estimate}~$\hat{w}_n$ obtained from~$n$ measurement outcomes can differ from the true value~$\Tr[W \rho^*]$.
Therefore, it is an important question how to quantify the confidence in the experimentally determined estimate~$\hat{w}_n$.

\subsection{Prior work and motivation}\label{subsec:prior}
In many experiments, the confidence in the estimate~$\hat{w}_n$ of the true witness expectation value~$\Tr[W \rho^*]$ is expressed by the standard error~$\hat{\sigma}$ (the empirical standard deviation)~\cite{Bourennane2004,Filgueiras2012,Song2017,Friis2018,Zhong2018,Wang2018,Humphreys2018,Dai2014}.
These experiments typically claim the certification of entanglement if the estimate~$\hat{w}_n$ is a number of~$\hat{\sigma}$'s below zero.
This approach is simple and pragmatic, but may suffer from statistical and practical challenges (see~\cite{Zhang2011} for similar objections to using this method for quantifying Bell violations).
We give a concrete example in \cref{subsection:broken} (see \cref{fig:fid_hist_broken}) where this approach could potentially be problematic.

Certification of entanglement by the number of sigma's of witness violation is most easily justified under the assumption that in each round~$i$ the state~$\rho_i$ is independent and identically distributed (iid assumption).
This is equivalent to assuming each round a fixed state~$\rho^*$ is produced.
Under this assumption, the estimate~$\hat{w}_n$ is considered a realization of a Gaussian random variable~$\hat{W}_n$ with mean $\Ebb[\hat{W}_n] = \Tr[W \rho^*]$ and standard deviation~$\sigma\sim\frac{1}{\sqrt{n}}$ (for sufficiently large~$n$).
This is justified by the central limit theorem.
The empirically obtained~$\hat{w}_n$ and~$\hat{\sigma}$ are appropriate estimates of the mean~$\Tr[W \rho^*]$ and standard deviation~$\sigma$.
Thus, the reported~$\hat{w}_n \pm \hat{\sigma}$ is a complete and accurate characterization of the distribution (and hence leads to meaningful confidence intervals etc.)

However, if the states~$\rho_1, \dots, \rho_n$ produced in an~$n$ round experiment are not iid, several difficulties may arise.
This starts with what is even to be estimated.
Most natural is to estimate the \emph{average witness value}
\begin{equation}
\label{eq:avg_witness}
\braket{W}_n := \frac{1}{n} \sum_{i=1}^n \Tr[\rho_i W] = \Tr[(\frac{1}{n} \sum_{i=1}^n \rho_i) W],
\end{equation}
which can also be interpreted as the witness expectation value of the average state~$\rho^* =  \frac{1}{n} \sum_{i=1}^n \rho_i$.
There are versions of the central limit theorem that relax the iid assumption. They can be used to argue that the estimate~$\hat{w}_n$ is still an observation of a Gaussian random variable~$\hat{W}_n$ with mean~$\Ebb[\hat{W}_n] = \braket{W}_n$ for sufficiently large~$n$ (Gaussian assumption).
But in practical experiments it is not always clear when these theorems can be applied, so that the convergence of~$\hat{W}_n$ to a Gaussian with mean~$\braket{W}_n$ is not guaranteed for any~$n$.
Moreover, under non-iid states~$\rho_i$ it is unclear whether the observed standard error~$\hat{\sigma}$ is an appropriate estimator of the true standard deviation~$\sigma$.
Finally, even if the central limit theorem applies, the convergence of~$\hat{W}_n$ to a normal distribution can be extremely slow in~$n$ (especially when the witness violation is large \cite{Zhang2011}), so that too small~$n$ will still cause~$\hat{W}_n$ not to be Gaussian.
Hence, in practice it can be difficult to justify the Gaussian assumption.

When the iid assumption (or more generally the Gaussian assumption) fails, several different problems can arise.
First, it can lead to overestimation of the confidence in the witness violation based on the observed data.
This happens when~$\Pr[\hat{W}_n < 0]$ is smaller under the true distribution of~$\hat{W}_n$ than under the estimated distribution, based on the observed~$\hat{w}_n$ and standard error~$\hat{\sigma}$ (i.e., based on the Gaussian assumption).
Second, the reported numbers~$\hat{w}_n \pm \hat{\sigma}$ lack rigorous interpretation.
The empirical standard deviation~$\hat{\sigma}$ may no longer be an appropriate estimate of the true standard deviation~$\sigma$. Moreover, the mean and standard deviation do not necessarily describe the distribution of~${W}_n$ completely (if~$\hat{W}_n$ is not Gaussian, the true distribution may depend on more than 2 free parameters).
Because of these two effects, the number of~$\hat{\sigma}$'s of witness violation in relation to the estimate~$\hat{w}_n$ will depend on the way~$\hat{W}_n$ fails to be Gaussian or on how~$\hat{\sigma}$ fails to estimate~$\sigma$.
This also makes the results between different experiments and physical platforms become incomparable, because the actual distribution of~${W}_n$ may be influenced by experimental parameters, such as the distribution of~$\rho_i$, measurement settings, hardware imperfections or the choice of witness. Hence the number of~$\hat{\sigma}$'s of violation may also be influenced by these experimental details.

Finally, we note that measurement noise (systematic errors) can also lead to overestimation of the confidence in the witness violation.
This is because any measurement noise leads to the imperfect implementation~$\tilde{W}$ of the witness~$W$.
In case that~$\braket{\tilde{W}}_n < \braket{W}_n$, this again leads to an overconfidence in the witness violation.
In fact, it can even happen that~$\braket{\tilde{W}}_n < 0$ while~$\braket{W}_n \geq 0$, leading to falsely concluding entanglement~\cite{Rosset2012}. This overconfidence persists independent of the number of samples~$n$ taken, since the error is systematic.

\subsection{Our contribution}
We propose a new method of carrying out and analyzing witness experiments that addresses all of the aforementioned issues.
This method applies without any assumption on the states produced by the experiment (i.e., they may be arbitrarily correlated). Moreover, it has a simple and clear interpretation, and yields figures of merit that are easily comparable between different experiments.
Finally, our method takes into account imperfections of the measurement device and random setting generation, avoiding systematic overestimation of confidence.

In our method, we view the source of the states as a black box that produces a quantum state~$\rho_i$ on demand. The source can produce multiple states sequentially. All of these states are modeled by random variables that can be arbitrarily distributed and which may depend arbitrarily on the history of the experiment.
That is, we allow the source to have memory.
We now model the experiment as a sequential process of~$i = 1,\dots,n$ rounds. In each round, a state~$\rho_i$ and a random measurement setting (determined by the decomposition of~$W$ into locally measurable observables) are requested. The appropriate measurement is performed on the state and the outcomes are recorded for data processing. In this model of the experiment, we additionally allow for arbitrary bounded-strength noise in the measurement devices and random measurement setting generator. That is, we assume that the noise in these devices is smaller than a quantified maximum amount.
Witnessing entanglement without any assumptions on the devices, an area known as self-testing \cite{Supic2018,Supic2019}, is based on Bell-type inequalities, which are typically tighter than witness inequalities (and therefore harder to violate experimentally).
Thus, our model is very general and has minimal assumption to be as widely applicable as possible for analyzing witness experiments.

The main contribution of this work is two different types of data analysis for witness experiments. Both methods are valid under these extremely general assumptions (in particular the state assumptions). In the first method, we quantify the confidence that the source has \emph{the capability to produce an entangled state} (i.e., a state outside~$\Sc$).
This means that we rigorously determine the confidence that~$\Tr[\rho_i W] < 0$ for at least one~$i$.
We do this by applying the framework of hypothesis testing, in which a null hypothesis is to be rejected based on the observed evidence in experiment.
In witness experiments, the null hypothesis is that the source only produces separable states (i.e.~$\forall i: \, \rho_i \in \Sc$).
To reject this null hypothesis means that at least one entangled state must have been produced by the source (i.e.~$\exists i: \, \rho_i \notin \Sc$).
The figure of merit to quantify the confidence in rejecting the null hypothesis is the \emph{$p$-value}.
Intuitively, the~$p$-value is the probability of obtaining data at least as extreme as the observed data in an experiment \emph{if the experiment were governed by the null hypothesis}, i.e., if the source was only able to produce separable states.
A small~$p$-value is then considered strong evidence against the null hypothesis: the observed data is very unlikely to be explained by a model that includes the null hypothesis.
If~$p$ is smaller than some significance level~$\alpha$, the null hypothesis is rejected and we conclude that entanglement must have been produced at least once with confidence~$1-\alpha$. We shall refer to this entire procedure as a \emph{witness rejection experiment} and the data analysis as the \emph{rejection analysis}. This method is different from the standard methods, in the sense that here we can make a statement about \emph{at least one} state, whereas typically one makes a statement about the \emph{average produced states} (e.g. when estimating the average witness value~$\braket{W}_n$ as defined in~\cref{eq:avg_witness}).
We emphasize that our method and analysis applies in complete generality to arbitrary witness experiments (with an arbitrary convex set of states~$\Sc$ and witness~$W$).
For concreteness we will focus in this work on entanglement witnessing, but see \cref{subsec:applicability} for other examples.

In the second method we aim to estimate the \emph{average witness value}~$\braket{W}_n$ and we provide a confidence interval around this estimate.
The main contribution of our confidence interval method is that it is valid without any assumptions on the produced states and therefore always applies.
We will refer to this method as the witness \emph{estimation method} and the data analysis as the \emph{estimation analysis}, since the objective here is to estimate~$\braket{W}_n$.
This method is generally applicable to estimate any Hermitian observable, not just witness operators (i.e., it is not necessary that there is a set~$\Sc$ such that \cref{eq:witness_definition} holds).
Thus our estimation method even applies to fidelity estimation and other estimation experiments.

The contributions of our work are presented in the following way.
First, we formulate the round-by-round witness experiment as an entanglement witness game, expand on this description and present a formal model that governs the experiment in \cref{sec:formulation_model}. In the model description we incorporate imperfections in the measurement device and random setting generation in a quantitative way.
Based on this model, we give a step-by-step description how to set the parameters and carry out the witness experiment in \cref{subsec:experiment}.
Then, we show how to calculate a witness correction from the quantification of the imperfect experimental devices (\cref{subsec:theorem_gamma}, \cref{thm:gamma}). 
It is used in both the rejection and estimation experiments to account for systematic device errors and prevent possible overconfidence in the experiment outcomes (rejection and estimation).
Next, we provide the main result to perform the rejection analysis. This is an easy-to-compute bound on the~$p$-value (\cref{subsec:theorem_pvalue}, \cref{thm:pvalue}). The bound is simply evaluated from the measurement outcomes.
By comparing this bound to a predetermined significance level~$\alpha$, we can determine whether the experiment rejects the null hypothesis with statistical significance. This allows us to rigorously conclude that the source has the capability to produce entangled states with confidence~$1-\alpha$.
Finally, we provide the main result to perform the estimation analysis. This is a direct method to compute and estimate and confidence interval for the average witness value~$\braket{W}_n$ (\cref{subsec:theorem_CI}, \cref{thm:CI}). The estimate and this confidence interval are also directly and easily computable from the observed measurement data.
We illustrate these contributions with several detailed numerical examples in \cref{sec:example}. Two of our examples are based on the simulation of Nitrogen Vacancy centers. The focus of these examples is to detect genuine multipartite entanglement between three qubits (i.e., not a convex combination of biseparable states, states separable over some bipartition of the three subsystems). Our third example (\cref{fig:fid_hist_broken}) shows an explicit case where the Gaussian assumption fails to be applicable and where our methods are still applicable.

The technical ingredients of this work are summarized as follows. Both results are obtained by viewing the witness experiment as a game~\cite{Chen2017}, similar to Bell tests being viewed as nonlocal games. This allows us to construct (super)martingale sequences and use a concentration inequality to upper bound the tail probabilities (we use Bentkus' inequality~\cite{Bentkus2004,Bentkus2006}, which is slightly tighter than the more commonly used Hoeffding-Azuma inequality~\cite{Elkouss2016}). This method is inspired by the analysis of Bell inequalities of Ref.~\cite{Elkouss2016}. By choosing the appropriate (super)martingale sequences, we obtain the~$p$-value bound for the rejection analysis and the confidence interval for the estimation analysis.

\subsection{Relation to other work}
In this work, we model the measurement noise as general as possible via the POVM formalism and determine a witness correction from this model using analytical methods to guarantee we never overestimate the confidence. Our measurement model can be viewed as a generalization of the model studied in Ref.~\cite{Rosset2012}, where imperfect qubit measurements are modeled by Bloch vectors that are misaligned by at most some fixed angle.
In Ref.~\cite{Rosset2012} a witness correction factor is computed under this more restricted noise model. However, they compute the witness correction via numerical optimization (see \cref{subsec:gamma_other_method} for why we opt for an analytical bound and how the witness correction factor can alternatively be calculated using numerical optimization for our noise model).

The witness rejection experiment and analysis is new for entanglement witness experiments, but is inspired by the use of this technique for testing local realism through nonlocal games~\cite{Elkouss2016}. We emphasize that this rejection method aims to rigorously certify that a machine \emph{has the capability of} producing entanglement. This is different than typical witness experiments in literature where the objective is to estimate the average witness value~\cite{Bourennane2004,Filgueiras2012,Song2017,Friis2018,Zhong2018,Wang2018,Humphreys2018,Dai2014}.
The estimation method we study here also aims to estimate the average witness value. The main difference is that most works implicitly assume that the states are iid (or at least that the estimator is Gaussian by some type of central limit theorem argument), whereas our work applies in the most general case with arbitrary noise on the state. This makes our method more generally applicable.

Closely related to the confidence interval we construct here, Ref.~\cite{Ruan2020} provides a method to construct a Bayesian credible interval for an estimate of the average fidelity of experimentally prepared states to a fixed entangled state (note that is equivalent to a particular choice of witness). The model is similar in the sense that the states can be arbitrarily correlated, but the estimation objective is different: the goal of Ref.~\cite{Ruan2020} is to estimate the average fidelity of the unmeasured states from the measurement of a subset of all available states.
Similarly, Ref.~\cite{Takeuchi2019} derives an efficient method to verify the production of graph states by measuring all but one copy of the state.
In contrast, we measure all available states and only aim to make a statement about all the created states (after the fact). The work in Ref.~\cite{Walter2014} is related to this by giving general lower bounds on the size of a credible regions for quantum parameter estimation.

An alternative method to estimate a property of a quantum system, is by using quantum state tomography to collect measurement data, estimate a figure of merit (fidelity or witness value) and determine a confidence interval~\cite{Faist2016}.
However, this typically requires more measurement data than partial state characterization since the complete state is reconstructed.

Finally, we mention that there is also a way of witnessing entanglement without the need to trust the measurement devices at all (measurement-device-independent entanglement witnessing, MDI-EW)~\cite{Branciard2013,Nawareg2015}.
This, however, requires each party to hold auxiliary local quantum states in each round and perform a joint measurement between the auxiliary quantum state and the quantum state under investigation.
This method has been implemented in an experiment under the iid assumption~\cite{Nawareg2015}.

\section{Formulation and model of witness experiments} \label{sec:formulation_model}
In this section, we will discuss the formulation and modeling of witness experiments.
We will start with a brief review of entanglement witness games as known in the literature in \cref{subsec:preliminaries}.
Next, we will generalize the game formulation to handle two additional things: (1)~ multiple terms in the decomposition of the witness operator may be inferred from a single measurement; and (2)~measurements are allowed to be implemented by arbitrary POVMs. We explain how to do this and introduce notation in \cref{subsec:measurement-settings}.
Finally, in \cref{subsec:model_informal} we give a complete description of the experimental model that underpins our experiment. This includes the characterization of noisy measurement and random setting generation devices.

\subsection{Entanglement witness games}\label{subsec:preliminaries}
In this section, we will recap entanglement witness games from the literature.
We will start from the assumption that a choice of witness~$W$ has been made.
The quantum system under investigation is decomposed into~$m$~subsystems on which local measurements can be performed  (e.g.,~$m=2$ for bipartite entanglement witnessing).
The witness operator then admits a decomposition into locally measurable observables of the form
\begin{equation}
\label{eq:witness-decomposition-settings}
W = cI + \sum_x w_x M_x^{(1)} \otimes \cdots \otimes M_x^{(m)},
\end{equation}
where each~$M_x^{(j)}$ is a locally measurable observable on subsystem~$j$ and where~$x$ runs over the terms in the decomposition. Note that such a decomposition is always possible. A decomposition is minimal if the number of terms over which~$x$ runs is minimal.
In practice, the locally measurable observables~$M_x^{(j)}$ will often be Pauli observables.
The decomposition in \cref{eq:witness-decomposition-settings} is chosen such that each locally measurable observable can be easily measured in the experiment.
Measurement of~$M_x^{(j)}$ yields one of the possible outcomes labeled by~$a_j$ (in the case of Pauli observables, the outcomes are simply~$\pm 1$). We shall denote the vector of all outcomes of the~$m$ subsystems as
\begin{equation}
\mathbf{a} = (a_1,\dots,a_m).
\end{equation}

With this decomposition, an entanglement witness experiment can be formulated as a game~\cite{Chen2017}.
This is similar to how Bell experiments are often formulated as nonlocal games.
See \cref{fig:game} for an illustration of an entanglement witness game.
The game consists of~$n$~rounds.
There are~$m$~players, one for each subsystem.
At the start of each round, each player receives a subsystem of a quantum state~$\rho_i$, as well as a random measurement setting~$X_i$ (we will use the conventional notation of writing random variables as capital letters and their realizations as lowercase letters).
This random setting~$X_i$ dictates which measurements the players should perform on their local subsystems [according to the decomposition \cref{eq:witness-decomposition-POVM}].
Hence, upon receiving~$X_i = x$ in round~$i$, each player~$j$ perform the local measurement labeled by~$x$ and~$j$.
They then report their respective outcomes~$\mathbf{a}$ to a referee, who assigns a score to the round according to
\begin{equation}\label{eq:score_function}
  s(x,\mathbf{a}) = - \frac{w_x}{p_x} \prod_{j=1}^{m} a_j,
\end{equation}
where~$p_x$ is the desired probability of realizing measurement setting~$X_i = x$.
A priori, any choice of~$p_x$ defines a valid game.
However, the choice of~$p_x$ has a significant influence on finite statistics in an experiment.
We suggest a reasonable choice in \cref{eq:px_multisetting} and discuss the issue further in \cref{subsubsec:discussion_px}.
The negative sign in \cref{eq:score_function} is added conform the common convention that games aim to maximize score.
The score can be interpreted as the contribution of round~$i$ to the witness value.
Note that the score itself is a random variable~$S_i := s(X_i, \mathbf{A}_i)$, since it is a function of the random measurement setting~$X_i$ and the random measurement outcomes~$\mathbf{A}_i$.
The score is constructed in such a way that the expected value of the score (in the ideal scenario with perfect measurements and randomness) satisfies
\begin{equation}\label{eq:expected_score_ideal}
  \Tr[\rho_i W] = c - \Ebb[S_i | \rho_i],
\end{equation}
for all possible states~$\rho_i$. Thus, the witness expectation value is affinely related to the expected score of each round.

\begin{figure}
	\centering
	\includegraphics[width=0.95\linewidth]{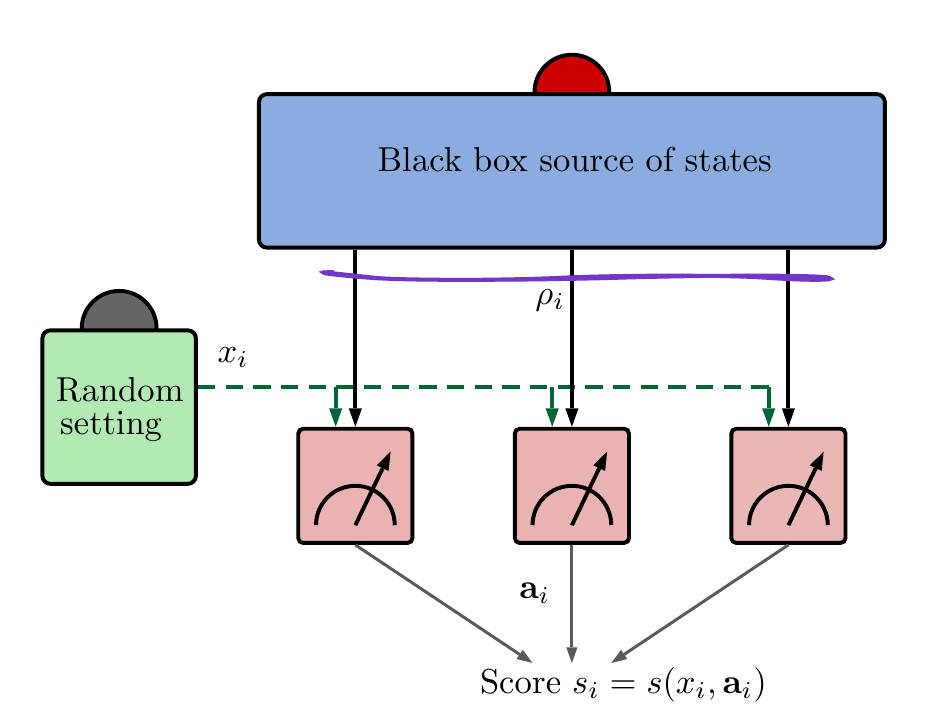}
	\caption{
  Round~$i$ of the entanglement witness game.
  By pressing the red button, a source produces a quantum state~$\rho_i$ and sends its subsystems to the players.
  We model the source as a black box, meaning it can produce a state by an arbitrary process.
  The state in round~$i$ can arbitrarily depend on everything that happened earlier in the experiment, i.e. the source is allowed to show memory effects.
  Then, by pressing the gray button, the random setting generator produces a measurement setting~$x_i$  (almost independently from~$\rho_i$).
  The players each perform a measurement according to setting~$x_i$ and send their outcomes~$\mathbf{a}_i$ to a referee, who computes the score of that round.
  Afterwards, round~$i+1$ starts.\label{fig:game}}
\end{figure}

\subsection{Generalization of the game formulation}\label{subsec:measurement-settings}
In this section, we will expand on the game formulation as discussed in the previous section. In particular, we will make two generalizations. First, we will explain and introduce notation to easily infer the expectation value of multiple terms in the witness decomposition \cref{eq:witness-decomposition-settings} from a single measurement. Doing this typically requires fewer measurements for fixed confidence so it is advantageous to do so when possible. Second, to be as general as possible in our measurement model, we shall allow every local measurement on a individual subsystem to be described by a POVM. Let us make these things more precise.

Sometimes it is not needed to measure all terms in \cref{eq:witness-decomposition-settings} separately~\cite{Bourennane2004,Toth2009}.
For example, with~$m=3$ single-qubit subsystems, Pauli-$Z$ measurements on each subsystem would allow to one infer the expectation values of all operators (omitting the tensor symbol)
\begin{equation}\label{eq:example_observables}
ZZZ, ZZI, ZIZ, IZZ, ZII, IZI, IIZ.
\end{equation}
This holds in general. Measurement of~$m$ non-identity operators on all of the subsystems, would allow one to infer~$2^m - 1$ expectation values.
We shall refer to the non-identity operators that are measured ($ZZZ$ in this example) as the \emph{measurement setting}~\cite{Bourennane2004,Toth2009} and refer to one or more of the possible~$2^m-1$ operators whose expectation value can be computed (operators from \cref{eq:example_observables} in this example) as \emph{observables}.
Throughout the rest of this work, we will denote measurement settings as~$M_x^{(1)} \otimes \cdots \otimes M_x^{(m)}$, where every~$M_x^{(j)} \neq I$ is not the identity, and index them with a subscript~$x$. We will denote observables as~$O_\xi^{(1)} \otimes \cdots \otimes O_\xi^{(m)}$ and index them with the subscript~$\xi$.
Note that~$\xi$ may run over different (more) terms that~$x$.
Using this new notation, the witness decomposition [\cref{eq:witness-decomposition-settings}] is thus written as
\begin{equation}\label{eq:witness-decomposition-observables}
W = cI + \sum_\xi w_\xi O_\xi^{(1)} \otimes \cdots \otimes O_\xi^{(m)}.
\end{equation}
To keep track of which observables (labeled by~$\xi$) are related to which measurement setting (labeled by~$x$), we define~$f(\xi) = x$ if the observable~$O_\xi^{(1)} \otimes \cdots \otimes O_\xi^{(m)}$ can be measured by the measurement setting~$M^{(1)}_x \otimes\cdots\otimes M^{(m)}_x$.
Furthermore, we define~$b(\xi) \in \{0,1\}^m$ as the bitstring of length~$m$ such that
\begin{equation}\label{eq:relation-O-M}
O_\xi^{(1)} \otimes \cdots \otimes O_\xi^{(m)} = (M_{f(\xi)}^{(1)})^{b_1(\xi)} \otimes \cdots \otimes (M_{f(\xi)}^{(m)})^{b_m(\xi)}.
\end{equation}
In this way, each term in \cref{eq:witness-decomposition-observables} is related to a measurement setting from which it can be obtained.
For example, the observables~$IZZ$,~$ZIZ$,~$ZZI$ can all be measured by the setting~$ZZZ$, and the corresponding bitstrings~$b$ are~$011$,~$101$,~$110$, respectively.
Note that if all observables require a different setting, then~$\xi = f(\xi) = x$, $b(\xi) = 11\cdots1$ and $O_\xi^{(j)} = M_x^{(j)}$, thus reducing to the simple case discussed in \cref{subsec:preliminaries}. Using this notation, we can write \cref{eq:witness-decomposition-observables} alternatively as
\begin{equation}\label{eq:witness-decomposition-POVM}
W = cI + \sum_\xi w_\xi \bigotimes_{j=1}^m \left( M_{f(\xi)}^{(j)} \right)^{b_j(\xi)}.
\end{equation}
To allow for the most general model of measurements, we will allow each~$M_x^{(j)}$ in a measurement setting to be measured by a POVM $\{\Pi^{(j),x}_a\}_{a \in \Omega^{(j)}_x}$ with outcomes labeled by~$a$ (which take values in the finite set $\Omega^{(j)}_x$). That is, we will write
\begin{equation}\label{eq:observable_decomposition}
M_x^{(j)} = \sum_{a \in \Omega^{(j)}_x} a \, \Pi^{(j),x}_{a}.
\end{equation}
For a standard measurement of the observable~$M_x^{(j)}$, this decomposition is simply given by the spectral decomposition, so that the POVM elements are the eigenprojections and the outcomes are simply the eigenvalues of $M_x^{(j)}$.
However, this is not the only option: the decomposition is not unique.
In particular, the POVM need not be a projective measurement.
This allows for the modeling of known non-unitary measurement noise.
Suppose for example that we wish to implement the measurement setting~$M_x^{(j)} = Z$, the Pauli $Z$-operator.
Its standard implementation would be by a projective measurement in the $\ket 0$, $\ket 1$-basis.
This corresponds to the decomposition $Z = \ketbra{0} - \ketbra{1}$ in \cref{eq:observable_decomposition}.
However, suppose that we can not perfectly discriminate~$\ket{0}$ from~$\ket{1}$ and you incorrectly obtain the opposite outcome in $1\%$ of the measurements.
Such a situation is modeled, e.g., by the POVM
$
\{ \begin{bsmallmatrix}
0.99 & 0 \\
0 & 0.01
\end{bsmallmatrix},
\begin{bsmallmatrix}
0.01 & 0 \\
0 & 0.99
\end{bsmallmatrix} \}.
$
Nevertheless, this POVM allows us to implement the desired measurement setting if we choose $a = \pm \frac{1}{0.98}$.
Indeed,
\begin{equation}
Z = \frac{1}{0.98}  \begin{bmatrix}
0.99 & 0 \\
0 & 0.01
\end{bmatrix} -
\frac{1}{0.98} \begin{bmatrix}
0.01 & 0 \\
0 & 0.99
\end{bmatrix}.
\end{equation}
Our results take into account the additional statistical uncertainty introduced by non-projective measurements in estimating the expectation value from finite single-shot measurements on the level of single-shot outcomes. Requiring that \cref{eq:observable_decomposition} holds, ensures the expectation value of this non-projective measurement equals the expectation value of observable to be implemented.

\begin{table}
	\caption{Summary of notation to allow multiple observables per measurement setting.
		The objects in the top half relate only to the choice of witness and its decomposition into observables (labeled by~$\xi$).
		The objects in the bottom half relate to the implementation of the witness using measurement settings (labeled by~$x$), which are implemented by POVMs.}
	\ra{1.5}
	\centering
	\begin{tabularx}{\linewidth}{L{0.44} L{0.50}  L{1}}
		\toprule
		Object & Symbol(s) & Definition / constraint 	\\
		\midrule
		Witness &$W$ &$\Tr[\rho W] \geq 0, \, \forall \rho \in \Sc$ \\
		Observable &$O_\xi^{(j)}$ & \multirow{2}{*}{$W = cI + \sum_\xi w_\xi \bigotimes_{j=1}^m O_\xi^{(j)}$}  \\
		Weight & 	$w_\xi, c$ &   \\
		\midrule
		Setting &$M_x^{(j)}, f, b$ &$O_\xi^{(j)} = (M_{f(\xi)}^{(j)})^{b_j(\xi)}$ \\
		POVM &$\{\Pi^{(j),x}_{a}\}_{a\in\Omega_x^{(j)}}$ &$M_x^{(j)} = \sum_{a\in\Omega_x^{(j)}} a \, \Pi^{(j),x}_{a}$ \\
		Distribution &$p_x$ & \cref{eq:px_multisetting} recommended \\
		\bottomrule
	\end{tabularx}
	\label{tab:notation}
\end{table}

With the generalizations just discussed, the score function in \cref{eq:score_function} needs to be generalized to
\begin{equation}\label{eq:score_function_xi}
s(x,\mathbf{a}) =  - \frac{1}{p_x} \sum_{\xi: f(\xi)=x}  w_\xi \prod_{j=1}^{m} (a_j)^{b_j(\xi)},
\end{equation}
The score now sums over all observables~$\xi$ obtained from the same setting~$x$.
The fact that the outcomes are raised to the power~$b_j(\xi)$ reflects the fact that~$O_\xi^{(j)} = I$ if~$b_j(\xi)=0$ (in which case all outcomes are $1$).
Note that the weights~$w_\xi$ are labeled by~$\xi$ as they appear in \cref{eq:witness-decomposition-observables}, whereas the probabilities~$p_x$ are labeled by the measurement setting~$x$.
This generalized score function still satisfies \cref{eq:expected_score_ideal} and is related to the witness decomposition \cref{eq:witness-decomposition-POVM} via (see \cref{app:game} for details)
\begin{equation}\label{eq:relation_W_to_score}
W = cI - \sum_x p_x \sum_{\mathbf{a}}  s(x, \mathbf{a})  (\Pi^{(1),x}_{a_1} \otimes \cdots \otimes \Pi^{(m),x}_{a_m}).
\end{equation}
We give an overview of the introduced notation in \cref{tab:notation}.

\subsection{Model of the experiment}\label{subsec:model_informal}
Above we explained that an entanglement witness experiment can naturally be interpreted as a game carried out by~$m$~players in~$n$~rounds (cf.~\cref{fig:game}).
We now summarize the key properties of our model in more detail -- see \cref{tab:model_main}.
A mathematically rigorous formulation is given in \cref{app:model}.

\Cref{assump_main:sequential} states that the experiment is performed sequentially.
Importantly, we do \emph{not} assume that the states~$\rho_i$ are independently and identically distributed (iid).
In fact, we allow that the~$i$-th round depends arbitrarily on the previous rounds.
Thus, the state~$\rho_i$ as well as the measurement setting~$X_i$ and the noisy POVMs elements of round~$i$ can be arbitrarily correlated and depend arbitrarily on the state, settings, POVMs, and outcomes of the previous rounds, as long as \cref{assump_main:randomness,assump_main:measurements} are satisfied.
This takes into account any possible systematic error in the experiment and in particular closes the detection loophole for entanglement witness experiments (the possibility of violating the witness due to classical correlations in POVM measurements)~\cite{Skwara2007}.

\Cref{assump_main:randomness,assump_main:measurements} model the devices used to perform the measurements in the experiment.
We assume that the random setting generator is characterized up to some precision~$\tau$ and that the measurement devices are characterized up to some~$\delta_j$, as defined in \cref{eq:tau_def,eq:delta_def} respectively.
In principle,~$\tau$ and the~$\delta_j$ may all depend on the round number~$i$, and the~$\delta_j$ may also depend on the measurement setting~$x$.
However, in practice, these dependencies will be small and we may safely take a maximum.
Moreover, it would be extremely impractical to characterize the devices for each round specifically.
The parameters~$\tau$ and~$\delta_j$ are later used to calculate the witness correction. With this we ensure that the confidence in the witness violation is never overestimated, even in the presence of systematic device errors.

Finally, in the rejection experiment, we also need to formalize the Null Hypothesis which we wish to reject. In the case of entanglement witnessing, the Null Hypothesis is that all states produced~$\rho_i$ are separable in every round~$i$. We formulate this more generally, by letting~$\Sc$ a convex subset of states (e.g. the separable states) such that~$W$ is a witness operator for~$\Sc$. This means that~$\Tr[W \rho] \geq 0$ for all~$\rho \in \Sc$. Then we can finally state the Null Hypothesis for~$\Sc$ mathematically as the following assumption:
\begin{enumerate}[label=($\Hc_0$), topsep=18pt, leftmargin=40pt, rightmargin=24pt]
	\item \label[hypothesis]{assump_main:states} \textbf{Null Hypothesis.} In every round~$i$, the source produces a state~$\rho_i \in \Sc$.
\end{enumerate}
This assumption is to be rejected with statistical confidence by the experiment, as we will describe in the next section.

\begin{table}
	\caption{List of Model Assumptions on the experimental devices and the nature of the experiment.
		These assumptions should plausibly hold in the real experiment. The validity of our results depends on these assumption holding.
		We give a mathematically rigorous definition of the model in \cref{app:model}.
	}   \label{tab:model_main}
	\rule{\linewidth}{2pt}
	\begin{center}
		Model Assumptions
	\end{center}\vspace{-6pt}
	\rule{\linewidth}{1pt}
	\begin{minipage}{0.85\linewidth}
		\vspace{6pt}
		\normalsize
		\begin{enumerate}[label=(\Roman*),leftmargin=9pt]
			\item \textbf{Sequentiality.}\label[assump]{assump_main:sequential}
			Rounds of the witness game are played sequentially.
			At the start of each round~$i$, each player~$j$ receives one part of a joint state~$\rho_i$ generated by the black box source, as well as a random measurement setting~$x_i$.
			Each player~$j$ performs a POVM measurement that depends on the setting~$x_i$, and reports the outcome~$a_j$ to a referee, who computes the score~$s_i$ of that round using \cref{eq:score_function_xi}.
			The next round~$i+1$ only starts after the referee received all players' measurement outcomes for round~$i$.
			The experiment is allowed to depend arbitrarily on the past.

			\item \textbf{Trusted randomness.}\label[assump]{assump_main:randomness}
			The random setting generator produces in each round~$i$ a random setting~$X_i$, whose distribution~$\tilde{p}_{i,x}$ (conditioned on the history of the experiment and the state produced) is close to the desired distribution~$p_x$:
				\begin{equation}
				\label{eq:tau_def}
				|\tilde{p}_{i,x} - p_x | \leq \tau \quad \forall i,x.
				\end{equation}
				We assume $\tau < p_x$ for all~$x$, so that each setting has nonzero probability of being realized.

			\item \textbf{Trusted measurements.}\label[assump]{assump_main:measurements}
			In each round~$i$, each player~$j$ performs a noisy POVM measurement~$\{\tilde{\Pi}^{(j)}_{i,a}\}_{a \in \Omega_{X_i}^{(j)}}$ that is close to the ideal POVM from \cref{eq:observable_decomposition}:
			\begin{equation}\label{eq:delta_def}
			\lVert \tilde{\Pi}^{(j)}_{i,a} - \Pi^{(j),X_i}_{a} \rVert_\infty \leq \delta_j \quad \forall i,x, j, a.
			\end{equation}
			The noisy measurements have the same outcomes~$a \in \Omega_{X_i}^{(j)}$ as the ideal measurements. Measurement outcomes follow Born's rule.
		\end{enumerate}
		\vspace{8pt}
	\end{minipage}
	\rule{\linewidth}{2pt}
\end{table}

\section{Results}\label{sec:test}
In this section we present the main results of our work.
In \cref{subsec:experiment} we start by giving a step-by-step description of our method and give an outline on how to apply the rejection analysis and estimation analysis.
Then, in \cref{subsec:theorem_gamma}, we compute the witness correction as a function of the model parameters that quantify the maximum device noise ($\tau$ and the ~$\delta_j$'s).
Next, in \cref{subsec:theorem_pvalue} we provide an easy-to-compute upper bound to the~$p$-value, which is used to perform the witness rejection experiment.
Finally, in \cref{subsec:theorem_CI} we state how to compute a confidence interval for the average witness value \cref{eq:avg_witness}.
These results apply in the presence of arbitrary, possibly correlated noise on the states.

\subsection{The design and analysis of an experiment }\label{subsec:experiment}
In this section, we detail all steps necessary to apply our framework -- from the design of the experiment to the analysis of the data.
In \cref{tab:summary_experiment} we summarize our method.
We now explain each step in detail.

\begin{table}
	\caption{Outline of our method for designing, performing and analyzing witness experiments. We assume that the experiments are guided by the Model Assumptions of \cref{tab:model_main} and that an appropriate operator~$W$ is fixed.  }
	\label{tab:summary_experiment}
	\rule{\linewidth}{2pt}
	\begin{center}
		Outline of experiment design and analysis
	\end{center}\vspace{-6pt}
	\rule{\linewidth}{1pt}
	\begin{minipage}{0.85\linewidth}
		\vspace{12pt}
		\normalsize
		\begin{enumerate}[leftmargin=*]
			\item \emph{Define the experiment}. Choose a
			\begin{enumerate}[label=\alph*., itemsep=0pt, topsep=0pt, leftmargin=16pt]
				\item decomposition of~$W$ of the form \cref{eq:witness-decomposition-POVM};
				\item measurement model of the form \cref{eq:observable_decomposition};
				\item probability~$p_x$ of measurement settings [e.g., using \cref{eq:px_multisetting}];
				\item number of rounds~$n$;
				\item significance level~$\alpha$.
			\end{enumerate}
			For a rejection experiment,~$W$ should be a witness for some set~$\Sc$. Choose the Null Hypothesis to be \cref{assump_main:states}.

			\item \emph{Characterize devices} w.r.t.~the~Model Assumptions.
			Determine suitable~$\tau$ and~$\delta_j$ [see \cref{eq:tau_def} and \cref{eq:delta_def}] for the hardware devices.
			From this \emph{compute the witness correction}~$\gamma$ using \cref{thm:gamma}, \cref{eq:gamma_1,eq:gamma_2}.

			\item \emph{Carry out the experiment.} In each round~$i$, record the obtained score~$s_i$ using \cref{eq:score_function_xi}.

			\item[4a.] \emph{Perform the rejection analysis.}
			From the recorded scores, compute the total normalized score~$t_n$ using \cref{eq:total_normalized_score}.
			Evaluate the upper bound~$p_\text{bound}(t_n, n, \gamma)$ using \cref{thm:pvalue}, \cref{eq:pvalue}.
			If~$p_\text{bound} \leq \alpha$, reject~\ref{assump_main:states} and conclude that the source is \emph{capable of producing states~$\rho \notin \Sc$} with confidence at least~$1-\alpha$.
			Otherwise, the test was inconclusive.

			\item[4b.] \emph{Perform the estimation analysis.} From the recorded scores, compute the witness estimate~$\hat{w}_n$ using \cref{eq:witness_estimate_main}. From the significance level~$\alpha$, compute the confidence interval radius~$\varepsilon$ using \cref{thm:CI}, \cref{eq:vareps_main}. Then, $\Ic = [\hat{w}_n - \varepsilon, \hat{w}_n + \varepsilon]$ is a~$(1-2\alpha$) two-sided confidence interval and $\Jc = (-\infty, \hat{w}_n + \varepsilon]$ is a~$(1-\alpha)$ one-sided confidence interval for the unknown quantity~$\braket{W}_n$ as defined in \cref{eq:avg_witness}.
		\end{enumerate}
		\vspace{8pt}
	\end{minipage}
	\rule{\linewidth}{2pt}
\end{table}

In this work, we assume that the specific observable~$W$ has been chosen. Of course, this choice is part of defining the entire experiment. For the rejection analysis,~$W$ should be a witness for some set~$\Sc$ [as defined by property \cref{eq:witness_definition}].
With entanglement witnessing in mind,~$\Sc$ is the convex set of separable states and~$W$ is some entanglement witness.
See \cref{subsubsec:witness_choice} for a discussion on how to choose a suitable~$W$ for witness experiments.

\paragraph*{Step 1. Define the experiment.}
With a choice of~$W$ fixed, we now choose a decomposition of the witness~$W$ as in \cref{eq:witness-decomposition-POVM} (such a decomposition is not unique). A good decomposition minimizes the number of terms, while keeping each term simple to measure.

Then we choose an ideal model for the implemented measurements by describing each measurement as a POVM~$\{\Pi^{(j),x}_{a}\}_{a\in\Omega_x^{(j)}}$ that satisfies \cref{eq:observable_decomposition}.
These POVMs should model the real implementation of the local measurements as accurately as possible (we will quantify the deviation of the real measurement devices in the second step).
Note that these POVMs can simply be projective measurements.

Next, we choose the desired probability distribution~$p_x$ of the random measurement settings in each round.
In principle, this can be chosen arbitrarily and the method will still work, but it has significant influence on the finite statistics of the experiment.
We propose to choose~$p_x$ as
\begin{equation}\label{eq:px_multisetting}
p_x = \frac{\sum_{\xi: f(\xi)=x} \lvert w_\xi\rvert}{\sum_{\xi} \lvert w_\xi\rvert}.
\end{equation}
Here~$w_\xi$ are the weights appearing in the decomposition~\cref{eq:witness-decomposition-POVM}.
This equation can be interpreted as choosing~$p_x$ proportional to the sum of absolute values of the weights~$\lvert w_\xi\rvert$ of all observables~$\xi$ that correspond to setting~$x$.
Hence, heavy-weight terms are measured more frequently to increase the precision of estimating that term.
See \cref{subsubsec:discussion_px} for a more detailed discussion on choice of~$p_x$.
The choices made so far define the score function \cref{eq:score_function_xi} that assigns a score to each round~$i$ of the game.

Finally, we fix the number of rounds~$n$ to play in the entanglement witness game, as well as a significance level~$\alpha$ (typical values are~$\alpha = 0.05, 0.01, 0.001$). In the rejection experiment, the significance level determines how small the observed~$p_\mathrm{bound}$ on the~$p$-value must be in order for us to reject \cref{assump_main:states}. In the estimation experiment, the significance level~$\alpha$ determines the confidence of the constructed confidence intervals around the estimate.
For the entanglement rejection experiment, it is important that all these parameters, especially~$\alpha$ and~$n$ are set before the experiment is carried out (see  \cref{subsec:pvalue_discussion}).

\paragraph*{Step 2. Characterize devices.}
In this step, we need to characterize the measurement devices that aim to implement the POVM elements of~\cref{eq:observable_decomposition}, and the random setting generator that aims to implement~$p_x$.
This characterization is done by determining suitable~$\tau$ and~$\delta_j$'s such that \cref{eq:tau_def,eq:delta_def} hold (ensuring that \cref{assump_main:randomness,assump_main:measurements} plausibly hold).
In practice, this process requires calibration and characterization of the real experimental devices.
From the numbers~$\tau$ and~$\delta_j$'s obtained in this characterization, we can compute the so-called \emph{witness correction}~$\gamma = \gamma_1 + \gamma_2$ using \cref{thm:gamma}, \cref{eq:gamma_1,eq:gamma_2}, in~\cref{subsec:theorem_gamma}.
When appropriate, one can use a first-order approximation for~$\gamma_2$ given in \cref{eq:gamma_2_approx}.
The witness correction~$\gamma$ is defined such that it bounds~$\| W - \bar{W}_i \|_\infty$, where $\bar{W}_i$ is the effectively implemented operator in round~$i$ and $W$ is the ideal target operator.
It is a function of the parameters $\tau$ and $\delta_j$, which quantify the device imperfections.
The witness correction~$\gamma$ is used to protect against the largest possible systematic error in the experiment under the Model Assumptions of \cref{tab:model_main}.

\paragraph*{Step 3. Carry out the experiment.}
Play~$n$ rounds of the witness game.
Each round~$i$, receive a state~$\rho_i$ from the source and measurement setting~$X_i=x$ from the random setting generator.
Then each subsystem~$j$ performs the POVM measurement~$\{ \tilde{\Pi}_{i, a}^{(j)} \}_{a \in \Omega_{x}^{(j)}}$ corresponding to setting~$x$ and obtains one of the possible outcomes labeled by~$a_j$.
Collect all the obtained outcomes~$\mathbf{A}_i = \mathbf{a}$, compute and the score~$s_i = s(x, \mathbf{a})$ using the score function in \cref{eq:score_function_xi} and record~$s_i$.
After the data collection has completed, one can do the analysis. We differentiate between the \emph{rejection analysis} and \emph{estimation analysis}. Both can be done using the same recorded data.

\paragraph*{Step 4a. The rejection analysis.}
After the data collection has completed, we can determine if the experiment successfully rejected the Null \cref{assump_main:states} with confidence~$1-\alpha$. To do so, we compute the \emph{total normalized score}~$t_n$, defined by
\begin{equation}
\label{eq:total_normalized_score}
t_n = \sum_{i=1}^{n} \frac{s_i - \smin}{\ds},
\end{equation}
where~$\ds := \smax - \smin$ and
\begin{equation}
\label{eq:score_max_min}
\smin := \min_{x, \mathbf{a}} s(x, \mathbf{a}), \quad
\smax := \max_{x, \mathbf{a}} s(x, \mathbf{a}),
\end{equation}
are the algebraic minimum and maximum value of the score function, respectively.
Note that~$t_n \in [0,n]$. This total normalized score is the our test statistic for the hypothesis test.
We can reject the Null Hypothesis if the~$p$-value is at most~$\alpha$. The~$p$-value is defined as the probability
\begin{equation}\label{eq:def p value}
p := \Pr\bigl[T_n \geq t_n \big| \Hc_0 \bigr]
\end{equation}
of obtaining a total normalized score $T_n$ under the Null \cref{assump_main:states} that is at least as large as the observed total normalized score~$t_n$.
To determine if~$p \leq \alpha$, we compute an upper bound~$p \leq p_\mathrm{bound}(t_n, n, \gamma)$ to the~$p$-value in \cref{thm:pvalue}, \cref{eq:pvalue}, and compare~$p_\mathrm{bound}$ to~$\alpha$.
If~$p_\text{bound} \leq \alpha$ then we can reject the Null \cref{assump_main:states} with confidence at least~$1-\alpha$. We can therefore conclude that at least one state~$\rho \not\in \Sc$ must have been produced and therefore the source has the \emph{capability of producing} such states. In the context where~$\Sc$ is the set of separable states, this is interpreted as concluding that the source is capable of producing entangled states.
This logical reasoning is only valid if all the Model \cref{assump_main:randomness,assump_main:measurements,assump_main:sequential} in \cref{tab:model_main} hold. If these fail then one may incorrectly reject~$\Hc_0$.

\paragraph*{Step 4b. The estimation analysis.}
From the collected data, we can also estimate the average witness value~$\braket{W}_n$ as defined in \cref{eq:avg_witness} and construct a rigorous confidence interval for around the estimate in the presence of arbitrary noise. The average witness value is estimated by the estimator
\begin{equation}\label{eq:witness_estimate_main}
\hat{w}_n = c - \frac{1}{n} \sum_{i=1}^{n} s_i.
\end{equation}
This estimator can, in the absence of noise on the measurements and random number generation, be seen as an unbiased estimator of~$\braket{W}_n$ by \cref{eq:expected_score_ideal}. In the presence of unknown noise, the bias of the estimate can only be bounded by the witness correction~$\gamma$ (see the discussion in \cref{subsec:theorem_gamma}).
Using~$\gamma$, we can compute the radius~$\varepsilon$ of the confidence interval using \cref{eq:vareps_main}.
By \cref{thm:CI}, the interval $\Ic(\hat{w}_n) = [\hat{w}_n - \varepsilon, \hat{w}_n + \varepsilon]$ is a~$(1-2\alpha)$ two-sided confidence interval and $\Jc(\hat{w}_n) = (-\infty, \hat{w}_n + \varepsilon]$ is a~$(1 - \alpha)$ one-sided confidence interval for $\braket{W}_n$.
If~$W$ is a witness for the set~$\Sc$ in the sense of \cref{eq:witness_definition} and if~$\hat{w}_n + \varepsilon < 0$, then one can conclude that $\braket{W}_n < 0$, meaning that \emph{on average} states outside~$\Sc$ must have been produced, i.e.
\begin{equation}
\rho^* = \frac{1}{n} \sum_{i=1}^{n} \rho_i   \not\in \Sc,
\end{equation}
with confidence at least~$1 - \alpha$.
We emphasize that the intervals~$\Ic, \Jc$ are corrected for systematic (measurement and random setting generation) errors within the Model Assumptions via the witness correction~$\gamma$ of \cref{thm:gamma} (since~$\varepsilon$ depends on~$\gamma$) and that it is statistically rigorous for arbitrary state noise.

\subsection{Computing the witness correction}\label{subsec:theorem_gamma}
In this section we present \cref{thm:gamma} to compute the witness correction~$\gamma$ as a function of the randomness and measurement imperfection parameters~$\tau, \delta_j$ determined in Step~2 of \cref{tab:summary_experiment}.
The imperfect implementation of the measurements and random number generator will lead to an effectively implemented operator 
\begin{equation}\label{eq:relation_Wnoise_to_score}
\bar W_i = cI - \sum_x \tilde{p}_{i,x} \sum_{\mathbf{a}}  s(x, \mathbf{a}) \,  \bar \Pi^x_{i,\mathbf a}.
\end{equation}
Here $\bar \Pi^x_{i,\mathbf a}$ is the expected implemented joint POVM in round~$i$, conditioned on the history of the experiment, the state produced, and the event that $X_i = x$ (which happens with probability~$\tilde{p}_{i,x}$).
See \cref{app:theorem1} for a precise definition.
Note that this effectively implemented operator~$\bar W_i$ is closely related by the ideal witness operator~$W$ by comparing to \cref{eq:relation_W_to_score}.
Indeed, the ideal random setting distribution $p_x$ is replaced with the implemented distribution~$\tilde{p}_{i,x}$ (which differ little by \cref{assump_main:randomness}) and the ideal POVM elements~$\Pi^{(1),x}_{a} \otimes \cdots \otimes \Pi^{(m),x}_{a}$ are replaced with the conditional expected implemented joint POVM elements~$\bar \Pi^x_{i,\mathbf a}$ (which differ little by \cref{assump_main:measurements}).
The witness correction~$\gamma$ we derive in \cref{thm:gamma} precisely captures how much $\bar{W}_i$ can deviate from $W$ within the Model Assumptions of~\cref{tab:model_main}.
\begin{theorem}\label{thm:gamma}
	Let~$W$ be a Hermitian operator [not necessarily a witness in the sense of \cref{eq:witness_definition}] with decomposition and ideal implementation given by \cref{eq:witness-decomposition-POVM,eq:observable_decomposition}.
	Suppose the experiment is modeled by the Model Assumptions in \cref{tab:model_main}.
	Define the effectively implemented operator $\bar{W}_i$ by \cref{eq:relation_Wnoise_to_score}.
	Then, in every round~$i$, 
	\begin{align}\label{eq:witness_correction}
	\| W - \bar{W}_i \|_\infty \leq \gamma,
	\end{align}
	where the \emph{witness correction}
	$\gamma := \gamma_1 + \gamma_2$
	is the sum of the \emph{random number generation correction}~$\gamma_1$ and the \emph{measurement correction}~$\gamma_2$ defined by
	\begin{align}\label{eq:gamma_1}
	\gamma_1 &:= \tau \sum_x  \max_{\mathbf{a}}\;\lvert s(x,\mathbf{a}) \rvert \\
	\label{eq:gamma_2}
	\gamma_2 &:= \sum_\xi \lvert w_\xi \rvert \sum_{j=1}^m  \Big(\prod_{k=1}^{j-1} ( \lambda_\xi^{(k)} + \epsilon^{(k)}_\xi ) \Big) \, \epsilon^{(j)}_\xi \! \prod_{k=j+1}^{m} \! \lambda_\xi^{(k)},
	\end{align}
	respectively, in terms of
	\begin{equation}\label{eq:eps_j_xi}
	\epsilon^{(j)}_\xi := b_j(\xi) \delta_j \sum_{a \in \Omega_{f(\xi)}^{(j)}} \lvert a \rvert
	\quad\text{and}\quad
	\lambda_\xi^{(j)} := \lVert O_\xi^{(j)} \rVert_\infty.
	\end{equation}
\end{theorem}

\noindent
The proof is given in \cref{app:theorem1}.
Let us first explain why we call the quantity~$\gamma$ the witness correction.
An important consequence of \Cref{eq:witness_correction} is that
\begin{equation}
\lvert \Tr[W \rho_i] - \Tr[\bar{W}_i \rho_i] \rvert \leq \gamma
\end{equation}
for all $\rho_i$. This means the witness inequality~$\Tr[W \rho_i] \geq 0$ implies that~$\Tr[\bar{W}_i \rho_i] \geq -\gamma$. Thus, if~$W$ is a witness, then the  operator~$\bar{W}_i + \gamma I$ is also a witness. Hence,~$\gamma$ is the witness correction in the sense that any effectively implemented witness~$\bar{W}_i$ corrected by~$\gamma$ is still a witness.
We emphasize that this result does not say anything about the effects of finite statistics, but is solely about the required correction of expectation values due to imperfect devices.
That is, the factor~$\gamma$ protects against potential systematic errors in an experiment.

The witness correction~$\gamma$ has two terms,~$\gamma_1$ and~$\gamma_2$.
The term~$\gamma_1$ quantifies the correction due to imperfect random number generation.
The constant~$\gamma_2$ quantifies the correction due to measurement errors.
Thus,~$\gamma$ can be interpreted as the total correction required if the witness~$W$ is implemented with noisy measurements and with an imperfect number generator.
Note that the choice of~$p_x$ influences the correction~$\gamma_1$, as the score function~\cref{eq:score_function_xi} depends on~$p_x$.

The measurement correction~$\gamma_2$ has a simple first-order approximation under the assumption that~$\lambda_\xi^{(j)} = 1$, making it easier to compute. This assumption means that all measurement operators have eigenvalues in the interval~$[-1,1]$ and is satisfied for example by all Pauli operators.
Then a first-order approximation for~$\gamma_2$ is
\begin{equation}\label{eq:gamma_2_approx}
\gamma_2 = \sum_\xi \lvert w_\xi \rvert \sum_{j=1}^m \epsilon^{(j)}_\xi + O(\epsilon^2),
\end{equation}
where~$\epsilon$ is a constant such that~$\epsilon_\xi^{(j)} \leq \epsilon$ for all~$\xi,j$.
Hence, this is a good approximation if~$\epsilon \ll 1$. This is typically the case when~$\delta_j \ll 1$, which means that the measurement devices have been well-characterized.
In \cref{subsec:gamma_other_method}, we discuss a possible alternative method for deriving~$\gamma$.

\subsection{Bound on the \texorpdfstring{$p$}{p}-value for witness rejection experiments}\label{subsec:theorem_pvalue}
In this section, we give the main result to perform the rejection analysis in \cref{thm:pvalue}. The theorem provides an easy-to-compute upper bound on the~$p$-value under the Null \Cref{assump_main:states}. Recall that the~$p$-value is the probability of observing a total normalized score~$T_n$ under the Null \cref{assump_main:states} that is at least as large as the observed total normalized score~$t_n$ in the experiment,
$
p = \Pr[T_n \geq t_n | \Hc_0].
$
If the~$p$-value is smaller than a previously chosen significance level~$\alpha$, then we may consider the~Null \cref{assump_main:states} to be statistically unlikely to explain the observed~$t_n$, and we may reject the model at significance level~$\alpha$. To determine if~$p \leq \alpha$, we put an upper bound~$p_{\mathrm{bound}}$ on~$p$ in \Cref{thm:pvalue}, which can be compared to the significance level~$\alpha$.

\begin{theorem}\label{thm:pvalue}
  Let~$W$ be a witness operator [satisfying \cref{eq:witness_definition}] for the set~$\Sc$ with decomposition and ideal implementation given by \cref{eq:witness-decomposition-POVM,eq:observable_decomposition}.
  Suppose that the experiment is governed by the~Model Assumptions of \cref{tab:model_main} and consider the Null~\cref{assump_main:states} with respect to~$\Sc$.
  Let~$t_n$ denote the observed total normalized score after~$n$ rounds in the experiment.
  Then, the~$p$-value as defined in \cref{eq:def p value} is upper-bounded by
	\begin{equation}\label{eq:pvalue}
	 p_\text{bound} := eF^\circ_{n,\beta}(t_n),
	\end{equation}
	where
	\begin{equation}\label{eq:Fcirc}
	F^\circ_{n,\beta}(x) :=
	F_{n,\beta}(\lfloor x \rfloor)^{1-(x - \lfloor x \rfloor)} F_{n,\beta}(\lfloor x \rfloor+1)^{x - \lfloor x \rfloor}
	\end{equation}
	is the log-linear interpolation of the survival function of a binomial distribution with parameters~$n$ and~$\beta$,
	\begin{equation}\label{eq:binom}
	F_{n,\beta}(k) = \sum_{l = k}^{n} \binom{n}{l} \beta^l (1-\beta)^{n-l},
	\end{equation}
	and where
	\begin{equation}\label{eq:alpha}
	\beta = \min\Bigl(1, \frac{c + \gamma - \smin}{\ds} \Bigr).
	\end{equation}
	Finally,~$\lfloor x \rfloor$ is the largest integer less than or equal to~$x$.
\end{theorem}

\noindent
We give a detailed proof of \cref{thm:pvalue} in \cref{app:proof2}.
We construct a supermartingale sequence from the total normalized scores up to each round~$i$.
We then apply Bentkus' inequality~\cite{Bentkus2004,Bentkus2006} (a concentration inequality for bounded difference supermartingale sequences, similar to, but tighter than, the Hoeffding-Azuma inequality) to obtain an upper bound for the~$p$-value.
Our proof is inspired by the approach of~\cite{Elkouss2016} to certify Bell violations.

\subsection{Confidence intervals for average witness estimation experiments}\label{subsec:theorem_CI}
In this section, we give the main result for the witness estimation analysis in \cref{thm:CI}. The theorem provides confidences interval for the average witness expectation value~$\braket{W}_n$ as defined in \cref{eq:avg_witness}. The point estimate for~$\braket{W}_n$ is given in \cref{eq:witness_estimate_main}, and is a function of the scores recorded in the experiment.
We construct a~$(1-2\alpha)$ two-sided confidence interval~$\Ic$ and a~$(1-\alpha)$ one-sided confidence interval~$\Jc$.

\begin{theorem}\label{thm:CI}
	Let~$W$ be a Hermitian operator [not necessarily a witness in the sense of \cref{eq:witness_definition}] with decomposition and ideal implementation given by \cref{eq:witness-decomposition-POVM,eq:observable_decomposition}.
	Suppose that the experiment is governed by the~Model Assumptions in \cref{tab:model_main}.
	Let~$\hat{W}_n$ denote the average witness estimate as defined in \cref{eq:witness_estimate_main}.
	Fix the significance level~$\alpha\in[0,1]$.
	If~$\alpha < e (\half)^{n}$, define~$\varepsilon = \ds$, otherwise define~$\varepsilon \in [\gamma, \gamma+\ds]$ implicitly via
	\begin{equation}
	\alpha = e F^\circ_{n,\half}\Bigl(\frac{n}{2}(1 + \frac{\varepsilon - \gamma}{\ds})\Bigr). \label{eq:vareps_main}
	\end{equation}
	Here~$\gamma$ is defined in \cref{eq:gamma_1,eq:gamma_2}, and~$F^\circ_{n,\beta}$ is defined in \cref{eq:Fcirc} (with~$\beta=\half$ and~$e \approx 2.72$).

	Then, the intervals
	\begin{align}
	\Ic(\hat{w}_n) &:= [\hat{w}_n - \varepsilon, \hat{w}_n + \varepsilon] \\
	\Jc(\hat{w}_n) &:= (-\infty , \hat{w}_n + \epsilon]
	\end{align}
	are a~$(1-2\alpha)$ two-sided and a~$(1 - \alpha)$ one-sided confidence interval respectively for the average witness value~$\braket{W}_n$ as defined in \cref{eq:avg_witness}. That is
	\begin{align}
	\Pr[\braket{W}_n \in \Ic(\hat{W}_n)] &\geq 1-2\alpha, \\
	\Pr[\braket{W}_n \in \Jc(\hat{W}_n)] &\geq 1-\alpha.
	\end{align}
\end{theorem}

\noindent
The proof of this theorem is given in \cref{app:CI}. The confidence interval is also based on the construction of a martingale sequence and the application of Bentkus' inequality. The techniques are very similar to the proof of \cref{thm:pvalue}. We chose to use Bentkus' inequality because it is tighter than the more standard Hoeffding-Azuma inequality \cite{Elkouss2016}. The radius of the interval~$\varepsilon$ is however slightly more difficult because it involves (numerically) solving \cref{eq:vareps_main}. See \cref{subsubsec:hoeffding} for a brief discussion on this.

\section{Examples and illustration}\label{sec:example}
\begin{figure*}
	\centering
	\subfloat[Singe click entanglement (SCE) between two NV~centers.]{
		\begin{minipage}{0.47\textwidth}
			\includegraphics[scale=0.55]{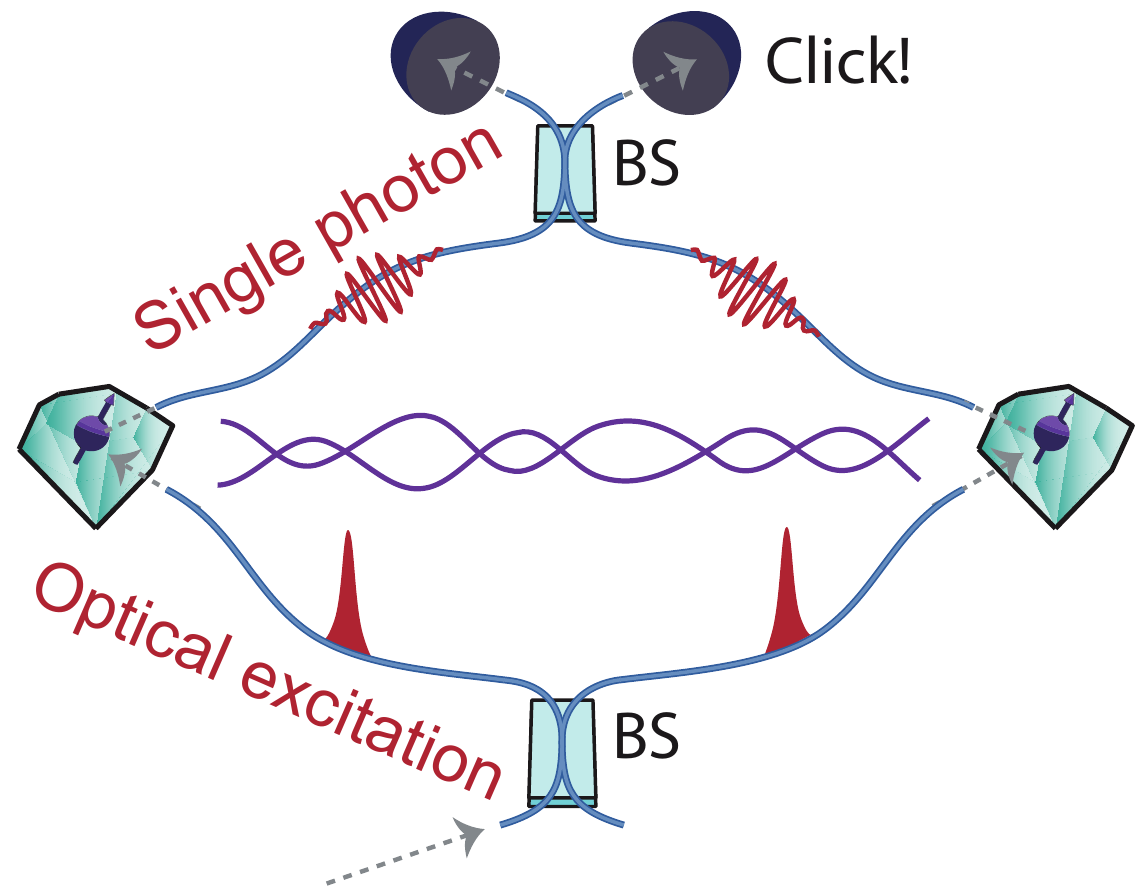}
		\end{minipage}
		\label{fig:NV_epr}
	}
	\hfill
	\subfloat[Schematic of tripartite GHZ generation.]{
		\begin{minipage}{0.47\textwidth}
			\includegraphics[scale=0.9]{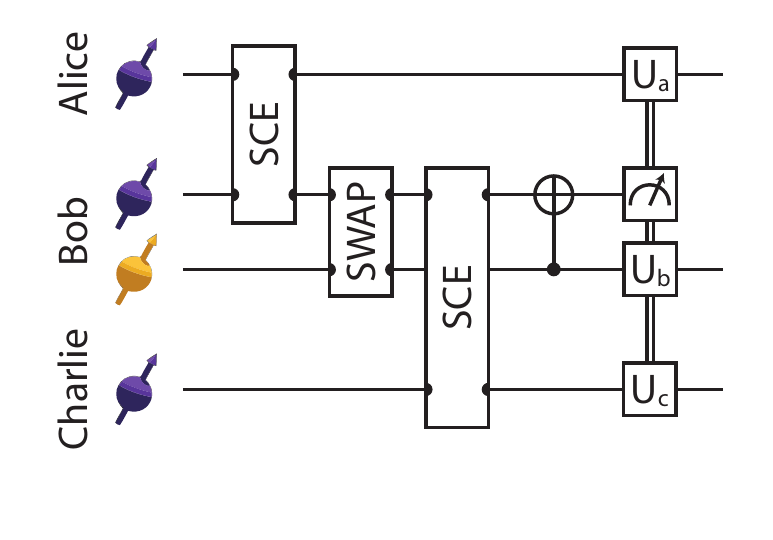}
		\end{minipage}
		\label{fig:NV_ghz}
	}
	\caption{Schematic illustration of tripartite entanglement generation using diamond Nitrogen Vacancy (NV) center systems.
		(a) The single click entanglement (SCE) scheme generates a single EPR pair between two NV centers.
		(b) Two EPR pairs are generated using the SCE scheme and combined into a GHZ state by interfering and measurement. The classically conditioned operations~$U_a$,~$U_b$ and~$U_c$ are Pauli operations.
	}
	\label{fig:NV}
\end{figure*}

In this section, we will illustrate our results with two examples based on simulations of a proposed entanglement witness experiment in Nitrogen Vacancy (NV) centers.
Moreover, we will give a concrete example in which the iid and Gaussian assumptions fail.
Finally, we shall illustrate how the function~$eF_{n,\beta}^\circ$ of \cref{eq:Fcirc} scales in its arguments and parameters. This function determines the~$p$-value bound and the confidence interval size in our results.
Before we present these examples, we briefly describe the physical system that we aim to simulate in \cref{subsec:NV}.
This section serves as a motivation for our simulation, but the examples can also be understood without knowledge of the physical system we simulate.
Then we present the two examples.
In the first example (\cref{subsec:example_steps}), we describe how to apply our method in detail, outlining all the steps in \cref{subsec:experiment}  in a concrete example.
For this example, we simulate a single experiment with identically distributed states~$\rho$.
In the second example (\cref{subsec:non-iid}), we illustrate our method for non-iid states.
To do so, we will perform a large Monte Carlo simulation of many independent experiments.
In each experiment, we use a sequence of three-qubit states~$\rho_i$ that are neither independent nor identically distributed.
Then, in \cref{subsection:broken}, we give an artificial example of non-iid states in which the Gaussian assumption fails considerably.
This example shows that a Gaussian assumption, on which the central limit theorem relies in prior work, need not always be justified (cf.~the discussion in \cref{subsec:prior}).
Our method applies regardless of the validity of a Gaussian assumption.
Finally, in the \cref{subsec:scaling} we illustrate how the function~$eF^\circ_{n, \beta}(x)$ defined in \cref{eq:Fcirc} [which directly determines the~$p$-value bound \cref{eq:pvalue}, and the confidence interval \cref{eq:vareps_main}] scales with~$n$,~$\beta$ and~$x$. Note that~$\beta$ scales linearly with the witness correction~$\gamma$ that captures device imperfections.

\subsection{Simulation details of Nitrogen Vacancy systems}\label{subsec:NV}
Both examples in \cref{subsec:example_steps,subsec:non-iid} are based on a scheme for generating tripartite GHZ~states in three physically separated Nitrogen Vacancy (NV) centers in diamond (see Ref.~\cite{Awschalom2018} for a review of this system).
In these NV~centers, the electronic spin associated with the defect can be used as qubit.
This qubit is optically accessible and can be entangled with the presence or absence of a photon, which can be used as a flying qubit.
Surrounding the NV center there are several Carbon-13 atoms (1.1\% natural abundance).
Their nuclear spins can be used as additional qubits, which can be controlled via the hyperfine interaction between the nuclear and electronic spins.

Two NV centers are entangled in the following way~\cite{Humphreys2018}:
First, each NV center produces a spin-photon entangled pair, where the qubit state is encoded in the absence/presence of a photon.
The joint state of the spin-photon pair is then given by
\begin{equation}\label{eq:z}
\sqrt{z} \ket{\uparrow}\ket{1} + \sqrt{1-z}\ket{\downarrow}\ket{0},
\end{equation}
where~$z$ is a tunable parameter.
By coupling the photons into single mode fibers and interfering them using a beam splitter, the two electronic spins can become entangled.
In essence, this amounts to detecting the presence of a photon but erasing the information about which arm the photon came from.
This \emph{single click entanglement (SCE)} scheme is illustrated in \cref{fig:NV_epr}.
The joint state between the two electronic spins is now (ideally)
\begin{equation}\label{eq:EPR_theta}
z \ketbra{\uparrow\uparrow} + (1-z)\ketbra{\Psi^+_\theta},
\end{equation}
where~$\ket{\Psi^+_\theta} = \ket{\uparrow\downarrow} + e^{i \theta} \ket{\downarrow\uparrow}$.
Here,~$\theta$ is a relative phase that needs to be characterized and controlled experimentally to create useful entanglement.

To generate a tripartite GHZ state, two EPR pairs are combined into a single GHZ state in the following way:
First create one EPR pair between two NV centers using the electronic qubits.
Then, one node swaps the state of the electronic spin with a nuclear spin qubit, so that the electronic spin becomes free again for entanglement production.
At this point a second EPR pair is produced between the now-free electronic spin of this node and a third node. 
The GHZ state is then created by coupling the nuclear spin and electronic spin in the middle node and measuring the electronic qubit.
This results into a state that is equivalent to a GHZ state under local Pauli operations (the Pauli operations can depend on the observed measurement outcome).
This procedure is sketched in \cref{fig:NV_ghz}.

When simulating this procedure, we account for several noise processes. First, we include noise in the generation of the EPR pairs. Our model for EPR generation follows the noise model for SCE generation developed in \cite{Humphreys2018}.
This model incorporates several independent noise parameters:
the single photon detection efficiency~$p_{det}$ (the probability of detecting a photon in the heralding station, conditioned on it being emitted from the NV center), the distinguishability~$V$ of the emitted photons,
the double excitation probability~$p_{2ph}$ (when more than one photon is emitted by an NV center in single entanglement attempt),
the probability of dark counts~$p_{dc}$,
as well as an uncertainty in the relative phase~$\theta$ that is  modeled by applying a Pauli-$Z$ to one of the qubits with probability~$p_{\theta}$.
We assume that the detection efficiency is the same for all three setups, and furthermore assume that in each SCE scheme symmetric values for the free parameter $z$ [see \cref{eq:z}] are used. However, this parameter may be different for the first and second EPR pair. We refer the reader to Ref. \cite{Humphreys2018} for full details on the SCE generation model.

On top of the detailed SCE noise, our model assumes dephasing noise on the first EPR pair while it is kept in memory, waiting for the successful generation of the second EPR pair. The off-diagonal terms in the density matrix of the first EPR pair are multiplied with dephasing parameter $q = 1 - \exp(-N_{\max} / \nu)$, where $N_{\max}$ is a free parameter determining the maximal number of attempts before we discard everything and start over (because the first EPR pair has decohered too much) and $\nu$ is a parameter quantifying the strength of the dephasing noise. Finally, we assume that all single- and two-qubit gates are performed with unit fidelity.
In \cref{subsec:example_steps} we instantiate this model with representative numerical values for all model parameters [see \cref{tab:simulation_params}] to produce the simulated experimental data.

\subsection{Step-by-step application of our method}\label{subsec:example_steps}
In our first example, we illustrate our method on data produced by a simulation of iid states on three qubits~($m=3$).
The aim is to witness a genuine tripartite entanglement by producing a GHZ~state:
\begin{equation}
\ket{\mathrm{GHZ}} = \frac{1}{\sqrt{2}} (\ket{000} + \ket{111}).
\end{equation}
Therefore, we let~$\Sc$ be the set of biseparable states, i.e., the set of all states~$\rho$ that are a mixture of separable states on any bipartition of the three subsystems.
To witness a state not in~$\Sc$ (and reasonable close to a GHZ state), we use is the projection witness, given by
\begin{equation}\label{eq:projection_witness}
W = \half I - \ketbra{\mathrm{GHZ}}.
\end{equation}
This factor~$\half$ is known to be optimal for the GHZ projection witness~\cite{Guhne2009}.
Note that~$c \neq \half$ in \cref{eq:witness-decomposition-observables}.
In fact, it is easily observed that $c = \half - \frac{1}{d} = \frac{3}{8}$ here (since $d = 2^m = 8$ in this example).
We will now describe all steps of \cref{subsec:experiment} to illustrate how to fully define the experiment, obtain the (simulated) data, and calculate the resulting~$p$-value and confidence interval of the experiment.

\paragraph*{Step 1a.}
The first step is choosing a decomposition of our choice of~$W$ [\cref{eq:projection_witness}] of the form \cref{eq:witness-decomposition-POVM}
This witness has a four-setting and five-setting decomposition.
We will use the five-setting decomposition into local Pauli observables:
\begin{multline}\label{eq:W_GHZ}
  W_\mathrm{GHZ} = \frac{1}{8} \Bigl( 3 III - IZZ - ZIZ - ZZI \\ - XXX + XYY + YXY + YYX \Bigr),
\end{multline}
where~$\{I,X,Y,Z\}$ are the Pauli operators (including the identity operator) and the tensor symbol is omitted for clarity.
Thus, \cref{eq:W_GHZ} is a decomposition of the form \cref{eq:witness-decomposition-observables} with~$c = \frac{3}{8}$.
We shall label the seven non-identity, traceless observables by~$\xi = 1,\dots,7$ in the order of their appearance in \cref{eq:W_GHZ}.

There are only five measurement settings needed for the decomposition in \cref{eq:W_GHZ}.
These are $\{ZZZ, XXX, XYY, YXY, YYX \}$, since the first three observables ($\xi=1,2,3$) can all be computed from the first measurement setting~$ZZZ$ ($x=1$), as discussed in \cref{subsec:measurement-settings}.
Therefore we have~$x = 1,\dots,5$, indexing the settings.
The formal mapping between observables and measurement settings is given by
\begin{equation}
f(\xi) =
\begin{cases}
1, & \mbox{if } \xi = 1,2,3, \\
\xi - 2, & \mbox{if } \xi = 4,5,6,7,
\end{cases}
\end{equation}
and
\begin{equation}
b(\xi) =
\begin{cases}
011, & \mbox{if } \xi = 1, \\
101, & \mbox{if } \xi = 2, \\
110, & \mbox{if } \xi = 3, \\
111, & \mbox{if } \xi = 4,5,6,7.
\end{cases}
\end{equation}

\paragraph*{Step 1b.}
Next, we specify a model for measuring each of the Pauli observables that occur in the chosen measurement settings. In our simulation, we shall model each measurement as systematically imperfect, meaning that~$X$,~$Y$ and~$Z$ is not implemented by the usual projective measurements.
Instead, we model all Pauli measurement by POVM elements, parameterized by two parameters~$u,v \in [0, 1]$, which characterize the efficiency of detecting the~$+1$ and~$-1$ eigenstate of the Pauli operator, respectively.
In experiments, these numbers are referred to as the readout fidelity \cite{Bultink2016,Humphreys2018}.
Concretely, for the Pauli-$Z$ measurement on each subsystem, we model (dropping the subsystem index~$j$ for notational compactness) the measurement by the POVM elements
\begin{equation}
  \Pi^Z_{a^+} = \begin{bmatrix}
  u  & 0 \\
  0 & 1 - v
  \end{bmatrix},
  \Pi^Z_{a^-} = \begin{bmatrix}
  1-u  & 0 \\
  0 & v
  \end{bmatrix}.
\end{equation}
The~$X$ and~$Y$ POVM elements are defined by
\begin{equation}
  \Pi^X_{a^\pm} = H \Pi^Z_{a^\pm} H^\dagger, \quad \Pi^X_{a^\pm} = K \Pi^Z_{a^\pm} K^\dagger,
\end{equation}
where the~$H$ and~$K$ gate are the gates that rotate the~$Z$ to the~$X$ and~$Y$ basis respectively.
The two outcomes of all Pauli measurements are
\begin{equation}
  a^\pm = \frac{v-u\pm 1}{u+v-1}.
\end{equation}
These values are chosen in such a way that
\begin{equation}
  a^+ \Pi^P_{a^+} + a^- \Pi^P_{a^-} = P,
\end{equation}
for all Pauli's~$P=X,Y,Z$, so that the measurement operators correspond to the desired observables [according to \cref{eq:observable_decomposition}].
In our model, we will set~$u = 0.95$ and~$v = 0.99$. This results in~$a^+ \approx 1.1064$ and~$a^- \approx -1.0213$.

\paragraph*{Step 1c.} Next, we choose~$p_x$ according to \cref{eq:px_multisetting}.
That is, we choose
\begin{equation}\label{eq:example_px}
p_x =
\begin{cases}
\frac{3}{7} &\mbox{if } x=1, \\
\frac{1}{7} &\mbox{if } x=2,3,4,5
\end{cases}
\end{equation}
This now fully defines the score function [per \cref{eq:score_function_xi}]:
\begin{multline}\label{eq:example_score}
s(x,\mathbf{a}) =  \begin{cases}
\frac{7}{24} (a_2 a_3 + a_1 a_3 + a_1 a_2 ), &  \mbox{if }x = 1, \\
\frac{7}{8} a_1 a_2 a_3, &  \mbox{if }x = 2, \\
-\frac{7}{8} a_1 a_2 a_3, &  \mbox{otherwise.} \\
\end{cases}
\end{multline}
Note how the observables~$IZZ$,~$ZIZ$ and~$ZZI$ are combined into one measurement setting~$ZZZ$ which directly contributes to the score.

\paragraph*{Step 1d-e.}Finally, we fix the total number of rounds to be~$n = 600$ and set~$\alpha = 0.05$.

\begin{table}
	\caption{Values of the parameters used in the simulated creation of a tripartite GHZ state in NV centers as discussed in \cref{subsec:NV}. The resulting state is described in \cref{tab:rho}.}
	\ra{1.07}
	\centering
	\begin{tabularx}{\linewidth}{p{0.22\linewidth} p{0.25\linewidth} | p{0.03\linewidth} p{0.22\linewidth} p{0.25\linewidth}}
		\toprule
		Noise Parameter & Value & & Free Parameter & Value	\\
		\midrule
		$p_{det}$ & $1.5 \cdot 10^{-3}$ & & $z_1$ & $0.016$ \\
		$V$ & $0.9$ & & $z_2$ & $0.080$ \\
		$p_{2ph}$ & $0.02$ & & $N_{\max}$ & $468$ \\
		$p_{dc}$ & $4.0 \cdot 10^{-7}$ & & & \\
		$p_{\theta}$ & $0.030$ & & & \\
		$\nu$ & $1500$ &  & & \\
		\bottomrule
	\end{tabularx}
	\label{tab:simulation_params}
\end{table}

\paragraph*{Step 2.}
We characterize our (simulated) devices to have a RNG bias~$\tau = 10^{-6}$ and measurement imperfection as compared to the model described in the previous step of~$\delta_j = \delta = 2 \cdot 10^{-3}$ for all parties~$j=1,2,3$. The value of $\delta$ is determined from the uncertainty in the measurement characterization of the NV system. The value of $\tau$ is chosen sufficiently large for any practical implementation of randomness.

With these values of~$\tau$ and~$\delta$, and the score function \cref{eq:example_score}, the witness correction~$\gamma$ can be computed from \cref{thm:gamma}.
The random number generation correction~$\gamma_1$ is computed using \cref{eq:gamma_1} to be
\begin{equation}
\gamma_1 = \tau \sum_{x=1}^5 \max_{\mathbf{a}} |s(x,\mathbf{a})| \approx 5.8 \cdot 10^{-6}.
\end{equation}
The measurement correction~$\gamma_2$ is computed using \cref{eq:gamma_2,eq:eps_j_xi}.
First, we compute~$\epsilon_\xi^{(j)}$ from \cref{eq:eps_j_xi}.
We find that
\begin{equation}
\epsilon_\xi^{(j)} = \begin{cases}
0, & \mbox{if } \xi=j=1,2,3, \\
\delta(|a^+| + |a^-|) \approx 4.26 \cdot 10^{-3}, & \mbox{otherwise.}
\end{cases}
\end{equation}
Furthermore, from \cref{eq:W_GHZ}, it is clear that $\lambda_\xi^{(j)} = 1$ for all~$\xi, j$, since all local observables $O_\xi^{(j)}$ are any of the four Pauli operators~$I,X,Y,Z$, which have operator norm~$1$.
Since~$\lambda_\xi^{(j)} = 1$ and~$\epsilon_\xi^{(j)} \ll 1$, we use the approximation \cref{eq:gamma_2_approx} instead of the full \cref{eq:gamma_2}.
We compute that
\begin{equation}
\gamma_2 \approx \sum_{\xi=1}^7  |w_\xi| \sum_{j=1}^3  \epsilon^{(j)}_\xi \approx 9.6 \cdot 10^{-3}.
\end{equation}
Hence we find that
\begin{align}\label{eq:concrete gamma}
\gamma = \gamma_1 + \gamma_2 \approx 9.6 \cdot 10^{-3} \leq 0.01.
\end{align}

\begin{table}
	\caption{Nonzero components of the state~$\rho$ used in each round of the simulation of the entanglement witness experiment described in \cref{subsec:example_steps}.
		The expected witness value of~$\rho$ is~$\Tr[W \rho] = -0.172$. }
	\ra{1.5}
	\centering
	\begin{tabularx}{\linewidth}{ C{1} C{1} | C{1} C{1} | C{1} C{1}}
		\toprule
		$M$    &~$\Tr[M \rho]$ 	&~$M$ 		&~$\Tr[M \rho]$	 		&~$M$ 		&~$\Tr[M \rho]$		 \\
		\midrule
		$III$  &~$1$   			&~$XXX$  	&~$0.782$  				&~$IIX$ 	&~$-0.077$			 \\
		$IZZ$  &~$0.787$ 		&~$XYY$ 	&~$-0.737$  			&~$XXI$ 	&~$ 0.072$			 \\
		$ZIZ$  &~$0.478$ 		&~$YXY$  	&~$-0.478$  			&~$YYI$  	&~$-0.047$  			 \\
		$ZZI$  &~$0.608$ 		&~$YYX$ 	&~$-0.507$  			&~$ZZX$ 	&~$-0.047$  			 \\
		\bottomrule
	\end{tabularx}
	\label{tab:rho}
\end{table}
\paragraph*{Step 3.}
In this step, the experiment is simulated.
We play~$n$ rounds of the entanglement witness game.
In our simulation, we take the same state~$\rho_i = \rho$ in each round, corresponding to an iid situation.
This state is computed using the model of tripartite GHZ generation in remote NV centers as discussed in \cref{subsec:NV}.
The valued of the parameters we used in the simulation are given in \cref{tab:simulation_params}.
The resulting state~$\rho$ is given in \cref{tab:rho}.
In addition to the state~$\rho_i = \rho$, we also generate a random setting~$x_i \in \{1,2,3,4,5\}$ with probability given by \cref{eq:example_px}. The randomness is generated by a standard pseudo-random number generator.
For efficiency reasons, we use the ideal POVM elements to simulate the measurement outcomes.
We nevertheless use a nonzero value of~$\delta$ to illustrate the effect of measurement noise on the witness correction.
The set of measurement outcomes~$\mathbf{a}_i$ is obtained according to Born's rule.
From this, we compute and record the score of this round~$s_i = s(x_i,\mathbf{a}_i)$ using the score function \cref{eq:example_score}.

\paragraph*{Step 4a.}
In this step, we calculate the upper bound~$p_\mathrm{bound}$ to the the~$p$-value using \cref{thm:pvalue}.
This is straightforwardly done using~$\gamma$,~$n$,~$c$,~$\smin$,~$\smax$ and the total normalized score~$t_n$ by evaluating \cref{eq:pvalue,eq:Fcirc,eq:binom,eq:alpha} of \cref{thm:pvalue}. The quantities~$\smin$ and~$\smax$ as defined in \cref{eq:score_max_min}, can be computed from the score function \cref{eq:example_score}. In our example we find
$
\smax = - \smin \approx 1.185.
$
The total normalized score~$t_n$ is computed from the recorded scores~$s_i$ for~$i=1,\dots,n$, according to \cref{eq:total_normalized_score}.
To compute~$p_\mathrm{bound}$ from concrete numbers, suppose the a single simulation of the experiment yields a total normalized score of~$t_n = 440.97$.
To compute the~$p$-value bound, we first calculate~$\beta$ using \cref{eq:alpha}. We find (using~$\gamma \leq 0.01$) that
$
\beta \approx 0.662.
$
Hence, we can evaluate our~$p$-value bound, \cref{eq:pvalue,eq:Fcirc}, by
\begin{equation}
p_{\mathrm{bound}} = e F_{n, \beta}(440)^{0.03} F_{n, \beta}(441)^{0.97} \approx 2.1 \cdot 10^{-4}.
\end{equation}
Since~$p \leq p_{\mathrm{bound}} \leq \alpha$, we have rejected the~Null \Cref{assump_main:states} at significance level~$\alpha = 5\%$.
As we trust that the Model \cref{assump_main:randomness,assump_main:measurements,assump_main:sequential} hold, we can reject \cref{assump_main:states} and conclude that our (simulated) source of quantum states is capable of producing genuinely multipartite entangled states~$\rho^* \notin \Sc$.
This is indeed the case, since~$\rho$ defined by \cref{tab:rho} is not biseparable (i.e., not in~$\Sc$, since~$\Tr[W \rho] < 0$).

\begin{figure}
	\centering
	\includegraphics[width=\linewidth]{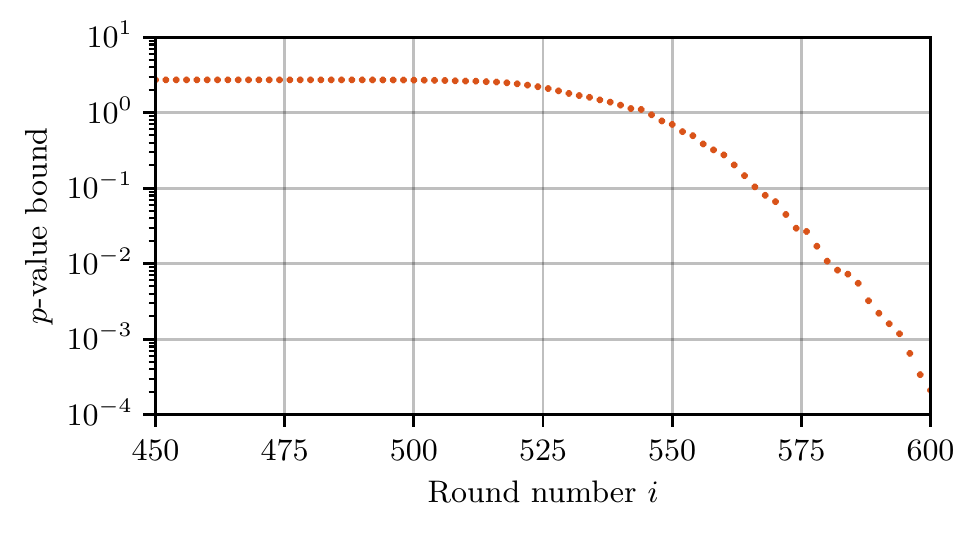}
	\caption{Running~$p$-value bound for our simulated GHZ witness experiment with parameters as defined in the main text. Up until round~$i=520$, the upper bound is the maximal value of~$e \approx 2.72$, due to the prefactor~$e$ in \cref{eq:pvalue}. The final~$p$-value bound after~$n=600$ rounds is~$2.1 \cdot 10^{-4}$.}
	\label{fig:running_pvalue}
\end{figure}

\begin{figure*}
	\centering
	\subfloat[Histogram of witness estimate]{
		\begin{minipage}{0.47\textwidth}
			\includegraphics[width=\linewidth]{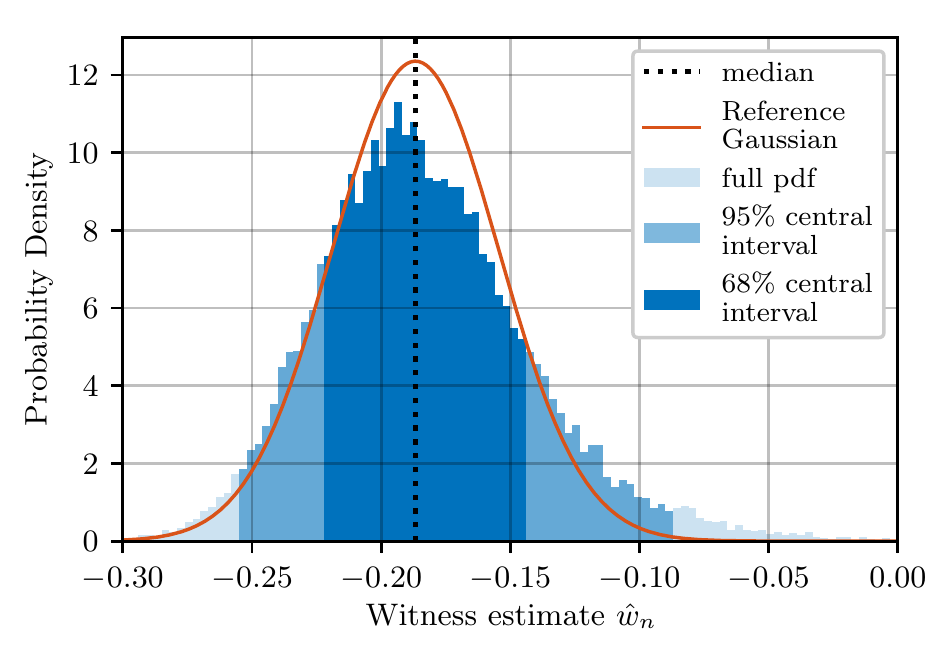}
			\vspace{-12pt}
			\label{fig:fidelity_pdf}
		\end{minipage}
		\quad
	}
	\hfill
	\subfloat[Histogram of~$p$-value bound]{
		\begin{minipage}{0.47\textwidth}
			\includegraphics[width=\linewidth]{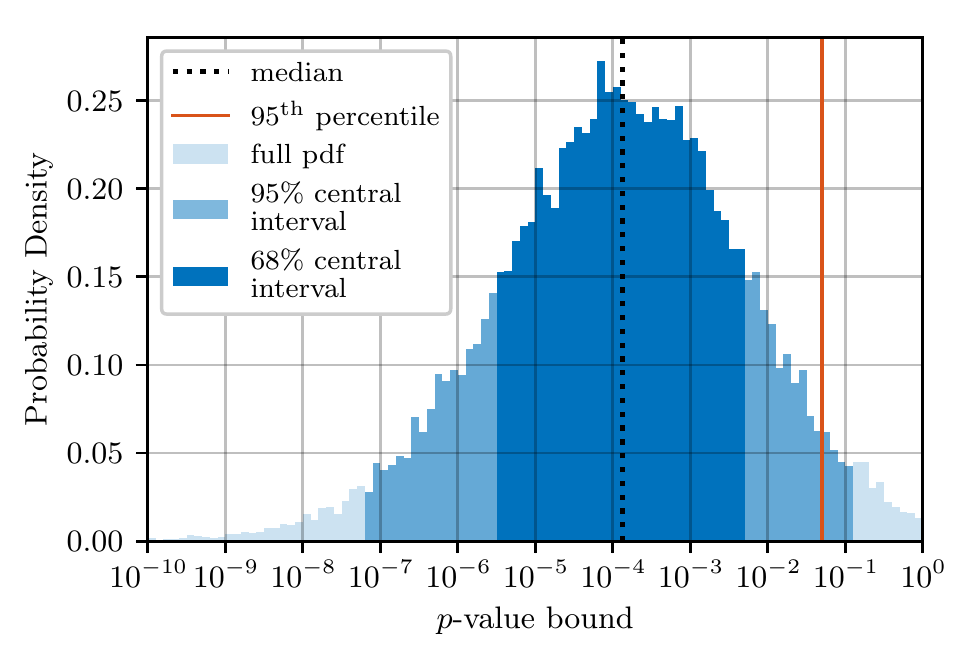}
			\vspace{-12pt}
			\label{fig:pvalue_pdf}
		\end{minipage}
	}
	\caption{Histograms for a Monte-Carlo simulation of~$N = 2 \cdot 10^4$ independent witness experiments (each with~$n = 600$ rounds).
		Here, the noise model is not iid but rather follows a stochastic process, as described in \cref{subsec:non-iid}.
		(a) Histogram of witness estimates.
		The blue histogram shows the distribution of the witness estimate~$\hat{w}_n$ (cf. \cref{eq:witness_estimate_main}).
		The distribution is non-Gaussian, as is obvious from comparison with the red Gaussian reference curve.
		Thus, the standard Gaussian prediction~$\hat{w}_n \pm \hat{\sigma}_{\hat{w}_n}$ does not accurately predict future repetitions of the experiment.
		The skew in the distribution is primarily due to skew in the true average witness value~$\braket{W}_n$ across different runs of the non-iid experiment (due to the changing sequence of~$n$ states).
		(b) Histogram of~$p$-value bounds.
		The blue histogram shows the distribution of~$p$-value upper bounds (cf. \cref{thm:pvalue}).
		The red line indicates the significance level~$\alpha = 0.05$, the number below which we reject the Null \cref{assump_main:states}.\label{fig:simulation_results}}
\end{figure*}

To illustrate how the~$p$-value evolves as more rounds are played (more data is taken), we plot in \cref{fig:running_pvalue} the running value of~$p_\mathrm{bound}$ computed with the total normalized score up to round~$i$, as a function of~$i$.
Up until round~$i \approx 520$, the bound remains constant at its maximal value~$e\approx 2.72$, due to the prefactor of~$e$ in \cref{eq:pvalue}.
In this regime there is insufficient data to draw any conclusions.
Then our~$p$-value bound starts decreasing roughly exponentially in~$i$.
The jitter can be explained by the randomness due to measurement settings and outcomes.

\paragraph*{Step 4b.}
Finally, we illustrate how to estimate~$\braket{W}_n$ as defined in \cref{eq:avg_witness} and compute the confidence intervals around this estimate using \cref{thm:CI}. The witness estimate is directly computed from the observed scores and~$c$ using \cref{eq:witness_estimate_main}. Suppose that a run of the experiment yielded scores such that~$\hat{w}_n = -0.182$ [this is consistent with~$t_n = 440.97$ by comparing \cref{eq:witness_estimate_main,eq:total_normalized_score}].
Now the radius of the confidence interval can be computed from~$\alpha=0.05$,~$n=600$,~$\gamma=0.01$ and~$\ds=2.370$ by solving \cref{eq:vareps_main} numerically. We find that~$\varepsilon \approx 0.216$ in this example. Hence, we find the~$90\%$ two-sided confidence interval and the~$95\%$ one-sided confidence interval
\begin{equation}
\Ic_{0.9} = [-0.398, 0.034], \quad \Jc_{0.95} = (-\infty, 0.034],
\end{equation}
respectively. This is just insufficient to conclude that \emph{on average}~$\rho$ was not in~$\Sc$ (i.e., genuinely multipartite entangled). However, we can compare these intervals to the true value~$\braket{W}_n = -0.172$ (see \cref{tab:rho} and recall that~$\rho_i = \rho$ in this example). Clearly~$\braket{W}_n \in \Ic_{0.9}$ and~$\braket{W}_n \in \Jc_{0.95}$. Moreover, the point estimate~$\hat{w}_n$ is not far off the true value.

\subsection{Illustration of method for correlated noise}\label{subsec:non-iid}
In the previous example we gave a step-by-step illustration of how our method is applied to data gathered in a single (simulated) experiment with identical states each round.
Since our method is applicable in general, we now illustrate its use for states~$\rho_i$ that are neither independent nor identically distributed.
As before, we will consider states~$\rho_i$ that are created from two EPR pairs, each of which has a relative phase~$\theta$ (see \cref{eq:EPR_theta} in \cref{subsec:NV}), but we now assume that the relative phase~$\theta$ drifts between subsequent rounds following a random walk.
That is, the relative phase in round~$i$ depends on its value in round~$i-1$, which creates correlation between the rounds.
This modeled phase drift is motivated by the observation in experiments that the relative phase~$\theta$ changes over time due to fluctuation in the optical path length \cite{Humphreys2018}.
The simulation uses the same parameters as in \cref{subsec:example_steps}, \cref{tab:simulation_params}, except that now~$p_\theta = 0$. Instead, in the first round, we model (both EPR pairs) with known initial phase $\theta_0 = 0^\circ$. Then, each round, the phase drifts with step size $\Delta \theta = \pm 0.98^\circ$ (for each EPR pair the drift is an independent random walk). This creates dependence and correlation between the rounds.
Furthermore, we use the same witness [\cref{eq:projection_witness}] and employ the same measurement model with witness correction~$\gamma \leq 0.01$ [\cref{eq:concrete gamma}] as in \cref{subsec:example_steps}.

Instead of performing a single simulation with this noise model, we perform a Monte Carlo simulation with~$N = 2 \cdot 10^4$ repetitions of the experiment, each with~$n=600$ rounds.
For each repetition, we calculate the witness estimate from the observed score using \cref{eq:witness_estimate_main}. We emphasize that in this simulation, each repetition also a different~$\braket{W}_n$ is realized, because the collection of~$\rho_1,\dots,\rho_n$ can be different every time.
We also compute a bound on the~$p$-value for each repetition of the experiment by using \cref{thm:pvalue}.
We thus obtain~$N = 2 \cdot 10^4$ witness estimates and bounds on the~$p$-value from our Monte-Carlo simulation.

In \cref{fig:simulation_results} we plot histograms of the witness estimates~$\hat{w}_n$ and~$p$-value bounds~$p_\mathrm{bound}$, both of which are directly computed from the \emph{observed scores}~$s_1,\dots,s_n$ (via \cref{eq:witness_estimate_main} and \cref{thm:pvalue}, respectively).
We emphasize that all~$s_i$ are realizations of a random variable~$S_i$, determined by the random measurement settings~$X_i$, the random measurement outcomes~$\mathbf{A}_i$ (following Born's rule), and the correlated random states~$\rho_i$ produced by the source.
The produced states were not necessarily biseparable  (i.e., not necessarily in~$\Sc$).

In \cref{fig:fidelity_pdf} we plot a histogram of the witness estimates~$\hat{w}_n$ and compare to a Gaussian reference curve.
We clearly see that this noise process can lead to a non-Gaussian witness distribution.
The skew in the plot is mainly due to the fact that different true average witness values~$\braket{W}_n = \frac{1}{n} \Tr[W \rho_i]$ are realized in each run of the experiment.
The figure illustrates that under non-iid noise, the standard Gaussian prediction~$\hat{w}_n \pm \hat{\sigma}_{\hat{w}_n}$ does not accurately predict future repetitions of the experiment.
We emphasize that the plot in \cref{fig:fidelity_pdf} can not be interpreted as the probability distribution of~$\hat{w}_n$ under a particular, fixed sequence of states~$\rho_1,\dots,\rho_n$ (because each Monte Carlo iterate used a different sequence of states), and hence no inference could be made from this plot about the uncertainty of~$\hat{w}_n$.

In \cref{fig:pvalue_pdf} we plot a histogram of the~$p$-value bounds.
Note that the~$p$-value is plotted logarithmically on the~$x$-axis, making the bins of unequal size.
We show the significance level~$\alpha = 0.05$ as a red line, which in this example turns out to be the~$95^\mathrm{th}$ percentile. Thus, in~$95\%$ of the simulated experiments, the~Null \cref{assump_main:states} was rejected with statistical confidence.
Note that it is merely a coincidence that~$\alpha = 0.05$ marks the~$95\%$ percentile of the distribution of~$p_\text{bound}$ over the~$N$ simulated experiments.
We emphasize that the~$p$-value itself [as defined in \cref{eq:def p value}] can not be inferred from the simulation results.
This is because the~$p$-value is a statement about the distribution of the total normalized score~$T_n$ assuming that the Null \cref{assump_main:states} holds.
But here the states produced are not necessarily in~$\Sc$, so \cref{assump_main:states} is violated. See \cref{subsec:pvalue_discussion} for more detailed discussion.

\subsection{Example where iid assumption fails}\label{subsection:broken}
\begin{figure}
	\centering
	\includegraphics[width=\linewidth]{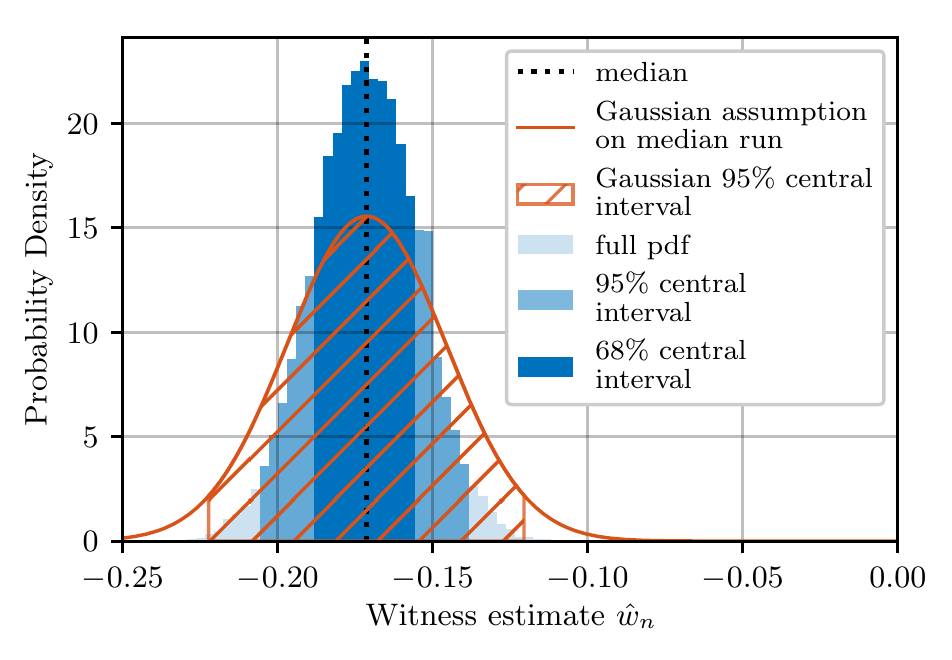}
	\caption{
		Histogram of a Monte-Carlo simulation of~$N = 2 \cdot 10^4$ witness experiments (each with~$n = 600$ rounds).
		The histogram shows the distribution of the witness estimate~$\hat{w}_n$ computed from the observed scores using \cref{eq:witness_estimate_main} in blue. This approximates the true distribution of~$\hat{w}_n$. The~$95\%$ interval is shaded in blue. The true average witness value is~$\braket{W}_n = -0.172$ in each run of the experiment, which coincides with the mean (and median)~$\hat{w}_n$.
		The red curve is shows the inferred distribution from the data (the scores) observed in the median run when the data is assumed to be iid. That is, the red curve is a Gaussian with mean~$\hat{w}_n$ and standard deviation~$\hat{\sigma}_{\hat{w}_n}$ as computed from the observed scores under the iid assumption in this particular run. The central~$95\%$ interval of this Gaussian is shaded with red lines.
		Clearly, the distribution inferred from the iid assumption does not match the true distribution of~$\hat{w}_n$.\label{fig:fid_hist_broken}
	}
\end{figure}

\begin{figure*}
	\centering
	\includegraphics[width=\linewidth]{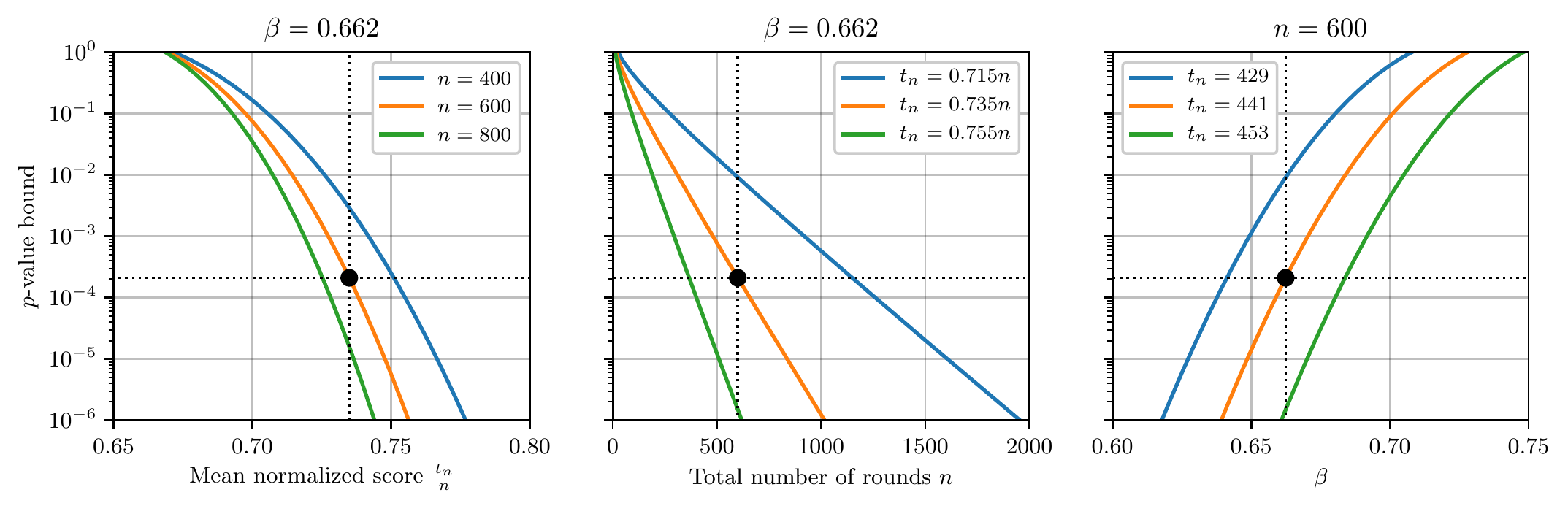}
	\caption{Scaling of the function $p_\mathrm{bound} = e F^\circ_{n, \beta}(t_n)$, which determines the bound on the $p$-value \cref{thm:pvalue} and the radius of the confidence interval \cref{thm:CI}, with the total normalized score~$t_n$~[\cref{eq:total_normalized_score}], the number of rounds~$n$ and~$\beta$~[which depends on $\gamma$, see \cref{eq:alpha}].
		The black dots correspond to the simulation discussed in \cref{subsec:example_steps}, where~$n=600$,~$\beta = 0.662$ (using~$\gamma = 0.01$), and~$t_n = 440.97$, yielding~$p_\mathrm{bound} = 2.1 \cdot 10^{-4}$. \label{fig:pvalue_scaling}}
\end{figure*}

In the previous section we saw an example where the states within each experiment were generated by a non-iid noise process, and we discussed the corresponding distribution of the witness estimate and~$p$-value bound across different runs of the non-iid experiment.
While the witness estimate~$\hat{w}_n$ was non-Gaussian, this was largely explained by skew in true average witness value~$\braket{W}_n$.
In this section we give an explicit example where the iid assumption (more generally the Gaussian assumption) is inappropriate even when the average witness  value~$\braket{W}_n$ is fixed. This example is not based on simulation of NV centers, but is based on a noise model that is designed to clearly exhibit non-iid behavior.

In this example, we generate the states according to a noise process in which the source of states produces a perfect GHZ state in a fraction~$F$ of the~$n$ rounds and produces an orthogonal state in the remaining~$(1-F)n$ rounds.
Importantly, the fraction of good states~$F$ is held constant in this model.
This is the only source of correlations between the states in a single run of the experiment.
This model is an extreme case of the more realistic scenario in which the source intermittently produces good and bad quality states. These bad quality states can for example be produced when a heralding signal (a signal that indicates entanglement was created) is falsely triggered.
Note that the fraction~$F$ of good states equals the true average fidelity to the GHZ state,
\begin{equation}
F = \frac{1}{n} \sum_{i=1}^n \Tr\bigl[\rho_i \ketbra{\mathrm{GHZ}}\bigr].
\end{equation}
and with our choice of witness operator~$W$ [according to \cref{eq:projection_witness}], the true witness value is
\begin{equation}
\braket{W}_n = \frac{1}{n} \sum_{i=1}^{n} \Tr[(\half I - \ketbra{\mathrm{GHZ}}) \rho_i] = \half - F.
\end{equation}
In this example, we fix~$F = 0.672$, which yields a true average witness value is~$\braket{W}_n = -0.172$.

In \cref{fig:fid_hist_broken} we plot a histogram of the witness estimate~$\hat{w}_n$, computed using \cref{eq:witness_estimate_main}, under the described noise model.
The histogram represents~$N = 2 \cdot 10^4$ independent simulated experiments (each with~$n = 600$ rounds).
We also plot the Gaussian distribution that would be obtained under the iid assumption on the median run.
For this run, the mean and standard deviation~$\hat{w}_n \pm \hat{\sigma}_{\hat{w}_n}$ are computed from the observed scores~$s_1,\dots,s_n$ in that run. In particular, for the red curve we observed~$\hat{w}_n \pm \hat{\sigma}_{\hat{w}_n} = -0.171 \pm 0.026$.
The~$95\%$ central interval is shown as the shaded region under the red curve. This region is $[\hat{w}_n - 2\hat{\sigma}, \hat{w}_n+2\hat{\sigma}] = [-0.223, -0.120]$.
However, the~$95\%$ central interval found from the Monte Carlo simulation is only $[-0.205, -0.138]$.
Hence, the~$95\%$ central interval estimated under iid assumption (red shaded region) is roughly~$1.5$ times larger than the~$95\%$ central interval found from the Monte Carlo simulation of~$\hat{w}_n$ (the medium blue region).
This example clearly shows that in this scenario the iid assumption fails.

\subsection{Scaling of the~\texorpdfstring{$p$}{p}-value bound}\label{subsec:scaling}
Finally, we analyze how our bound~$p_{\mathrm{bound}} = e F^\circ_{n, \beta}(t_n)$ [see \cref{thm:pvalue}, \cref{eq:pvalue}] scales with its three free parameters: the total normalized score~$t_n$, the number~$\beta$ (which depends on the witness correction~$\gamma$, see \cref{eq:alpha}) and the number of rounds~$n$.
Note that this also directly illustrates how the size of the confidence interval~$\epsilon$ scales with the significance level $\alpha$ according to $
\alpha = e F^\circ_{n,\half}(\frac{n}{2}(1 + \frac{\varepsilon - \gamma}{\ds}))$ [see \cref{thm:CI}, \cref{eq:vareps_main}].
We discuss the scaling here from the perspective of the rejection analysis.
Then it is qualitatively clear that a larger observed score~$t_n$ and smaller~$\gamma$ (and hence smaller~$\beta$) should lead to statistically more significant results (meaning a smaller $p_{\mathrm{bound}}$).
Similarly, a larger number of rounds~$n$ should lead to a smaller uncertainty and hence make it less likely to observe extreme data under the Null \cref{assump_main:states}.
We show that this is indeed the case in \cref{fig:pvalue_scaling}.
To compute~$p_{\mathrm{bound}}$ numerically, we used~$c = \frac{3}{8}$ and~$\smax = -\smin = 1.185$, just as in the example of \cref{subsec:example_steps} as the fixed values.
The black dots in the figure correspond to~$n=600$,~$\beta = 0.662$ and~$t_n = 440.97$ as used for illustration in \cref{subsec:example_steps}. The corresponding bound is~$p_\mathrm{bound} = 2.1 \cdot 10^{-4}$.

In the left plot of \cref{fig:pvalue_scaling}, we see that it appears that~$p_\mathrm{bound} \propto e^{-t_n^2}$ for fixed~$n$ in the regime $t_n \geq n\beta$.
Even if this is not exact,~$p_\mathrm{bound}$ at least appears to decrease (super)exponentially in~$t_n$.
The $p$-value bound is trivial in the regime $t_n \leq n \beta$. This is understood to mean that $n\beta$ is a total normalized score that is likely to be achieved under the null hypothesis.
In the middle plot, we see that~$p_\mathrm{bound}$ seems to decrease exponentially with~$n$ as well and that the asymptotic decay rate depends on the mean normalized score~$\frac{t_n}{n}$.
This is also expected behavior, which for example also holds in the iid case.
In the right plot, we focus on the effects of increasing~$\beta$.
The~$p$-value bound increases (sub)exponentially with~$\beta$.
It also seems that an increment of~$\beta$ (which is due to increased~$\gamma$ from noisy devices) can be directly compensated by observing an increased~$t_n$, since the three curves seem to be identical but displaced in the $x$-axis.


\section{Discussion}

\subsection{Applicability of our work} \label{subsec:applicability}
The methods we have discussed throughout this work were frequently motivated by and illustrated with entanglement witness experiments. Here we elaborate how generally applicable all of our results are.
First, the estimation experiment can in principle be done for any Hermitian observable~$W$. In \cref{thm:gamma,thm:CI} we do not require that the operator is a witness; we only need to write~$W$ in the form of \cref{eq:witness-decomposition-POVM,eq:observable_decomposition}. Note that this is in principle always possible (although for general~$W$, the number of terms labeled by~$\xi$ might grow exponentially with the total system dimension). This allows one to perform the experiment and do the estimation analysis of Step 4b. in \cref{tab:summary_experiment}. This yields a point estimate and confidence interval for estimating the average value~$\braket{W}_n$ as defined in \cref{eq:avg_witness}, which is valid without any assumption about the sequence of states~$\rho_1,\dots\rho_n$ produced sequentially in the experiment.

Second, the rejection analysis of the experiment is valid for any witness~$W$ that satisfies \cref{eq:witness_definition} for some convex subset of states~$\Sc$. The Null Hypothesis we aim to reject is always the hypothesis that the source can only produce states~$\rho \in \Sc$, conform \cref{assump_main:states}. Therefore, the experiment can witness the creation of any state~$\rho^* \not\in \Sc$, given suitable choices of~$W$ and~$\Sc$.
Examples include witnessing different types of entanglement, defined by non-membership of a particular separability class~$\Sc$.
A frequently encountered separability class is the set~$\Sc_k$ of~$k$-separable states.
These states are a convex combination of pure states that are separable over some~$k$-partition of the subsystems~\cite{Walter2016}.
Our examples in \cref{sec:example} focused on~$k = 2$, the class of biseparable states.
In certain applications it might also be important to certify entanglement across a particular partition of subsystems, which leads to a different definition of~$\Sc$.
In this case, the set~$\Sc$ describes separability between any specific partitioning of the system.
Corresponding witness operators for graph states have been identified in~\cite{Zhou2019}.
Finally,~$\Sc$ can also be chosen to include specific classes of entangled states, so that non-membership signals that the entanglement is not of that particular class. As an example of this, $\Sc$ can be chosen as the convex hull of all biseparable states and~$W$-class states~\cite{Walter2016}.
All of this can be tested with our method by accordingly defining~$\Sc$ in \cref{assump_main:states}.
Of course, this also requires one to find a suitable witness operator~$W$ that separates~$\Sc$ from some particular state~$\rho^*$ that one aims to witness [according to~\cref{eq:witness_definition}].

\subsection{Choosing the free parameters of the experiment}
\subsubsection{The witness operator~\texorpdfstring{$W$}{W}} \label{subsubsec:witness_choice}
The choice of witness operator is another integral part of the design of the experiment.
It affects the number of trials~$n$ needed in an experiment to attain a certain~$p$-value given some target state~$\rho^*$.
It affects the (expected) $p$-value bound or interval size $\varepsilon$ one would observe in an experiment given finite~$n$.
There are two different mechanisms behind this.
First, the choice of $W$ influences the decomposition \cref{eq:witness-decomposition-observables}, and in particular the number of terms that appear in this decomposition.
A smaller number of terms directly implies more statistically significant results.
This is intuitively understood since the budget of $n$ data points can now be used to measure less observables in the decomposition, meaning each can be estimated more accurately.
Second, for witness rejection experiments where some (average) $\rho^*$ is assumed to be produced, one would ideally like to minimize the $p_\mathrm{bound}$ over all witness operators~$W$ for fixed~$\rho^*$ and~$n$.
This task is extremely difficult and impractical.
Instead, one often employs the heuristic of minimizing the witness violation.
That is, one seeks to find the minimizer~$W$ of the following optimization problem
\begin{equation}
\min_{\substack{W : \|W\|_\infty \leq 1, \\ \Tr[W\rho]\geq 0 \; \forall \rho \in \Sc}} \Tr[W \rho^*].
\end{equation}
This witness will then give large contributions to the witness violation, hence making it statistically less likely to observe such a violation with finite number of states from~$\Sc$.
The downside of this is that it only optimizes for the magnitude of the violation in case of infinite statistics.
In particular, it does not take into account the number of terms in the decomposition.

Many different methods for entanglement witness design have been discussed extensively in literature~\cite{Guhne2009,Ryu2012,Toth2009}.
Most methods rely on some form of numerical optimization to find a witness with particular attractive properties such as efficient decomposition, maximal violation with a given state or large noise tolerance.
For pure states, the most widely used witness operator for genuine multipartite entanglement of states close to~$\ket{\psi}$ is the projection witness~$W = \lambda^2 I - \ketbra{\psi}$, where~$\lambda$ is the maximum of the Schmidt coefficients of~$\ket{\psi}$ when all bipartitions are considered~\cite{Toth2009,Guhne2009}.
This is also the choice that we made in \cref{sec:example} for~$\ket{\psi} = \ket{\mathrm{GHZ}}$ [see \cref{eq:projection_witness,eq:W_GHZ}].

\subsubsection{The number of rounds~\texorpdfstring{$n$}{n}}
An important design parameter of a witness experiment is the number~$n$ of rounds to be played.
This is often delicate, as too small~$n$ renders the entire experiment yielding no conclusion (see \cref{fig:running_pvalue}) but too large~$n$ can be overly costly in experiment resources.
To get a good estimate of~$n$ that produces small enough~$p_{\mathrm{bound}}$ without being unnecessarily expensive, the experimenter must have a good understanding of the quantum states being produced and the measurements being performed (e.g., from a theoretical model of the experiment).
With this, the experiment can be simulated, so that one can get an estimate of the distribution of the total normalized score~$T_n$.
Then pick~$n$ such that you are very likely (e.g., in~$90\%$ of the experiments) to realize a~$t_n$ that evaluates with known or estimated~$\gamma$ to a~$p_{\mathrm{bound}}$ less than~$\alpha$ (for the rejection experiment) or that evaluates given~$\alpha$ to a confidence interval that is sufficiently small (e.g. such that the entire interval is smaller than zero).

\subsubsection{The target distribution of measurement settings~\texorpdfstring{$p_x$}{px}} \label{subsubsec:discussion_px}
In principle, the choice of probability distribution over the measurement settings~$p_x$ is left free in this work.
Each choice leads to a valid witness experiment from which the conclusion can be drawn rigorously.
However, bad choices of~$p_x$ will lead to suboptimal bounds on the~$p$-value or suboptimal confidence intervals for equal experimental cost.
An interesting question is what choice of~$p_x$ yields the smallest~$p$-value bound (resp. confidence interval) given fixed~$n$ and~$\gamma$.
This question is hard to address, since~$p_x$ influences both the distribution of~$T_n$ and the value of~$\ds$ (which enters in \cref{thm:pvalue} via~$\beta$ and \cref{thm:CI} via \cref{eq:vareps_main}).
Based on the underlying concentration inequality (Bentkus' inequality, see \cref{lemma:bentkus} in \cref{app:proof2} for details), we conjecture that~$p_x$ should optimally be chosen to minimize~$\ds$.
This conjecture holds in the simulation of our experiments.
Intuitively, this requirement ensures that~$S_i - S_{i-1}$ is often large compared to~$\ds$.
However, finding a distribution~$p_x$ that minimizes~$\ds$ is not a simple problem.
In small instances, it can be solved by brute force search.
If all possible outcome sets are~$\Omega_x^{(j)} = \{1, -1\}$, then the~$p_x$ that minimizes~$\ds$ is given by our recommended choice in \cref{eq:px_multisetting}.

\subsection{Possible modifications of our method}
\subsubsection{Model modifications}
It is not hard to incorporate different ways of characterization the hardware devices [modifying \cref{assump_main:randomness,assump_main:measurements}].
For example, the bias in random measurement setting generation can be quantified by different~$\ell^q$-norms than the~$\ell^\infty$-norm considered here.
Similarly, there are many other ways of characterizing noisy measurement devices.
For example, one may model and characterize noisy measurements by a small misalignment angle, defined as the angle between the Bloch vectors of the ideal and the noisy measurements~\cite{Rosset2012}.
Any modification of the way hardware devices are characterized, requires that \cref{thm:gamma} is modified accordingly to calculate an appropriate correction~$\gamma$.
We emphasize that this is the only part that needs modification.
\Cref{thm:pvalue,thm:CI} can directly be applied with the new bound~$\gamma$.

\subsubsection{Alternative method for computing~\texorpdfstring{$\gamma$}{gamma}} \label{subsec:gamma_other_method}
In the proof of \cref{thm:gamma}, we applied an analytical method to find the witness correction~$\gamma$.
This bound is not necessarily tight in all cases.
If the correction~$\gamma$ is too large for a practical experiment, one can see if tigher bounds can be found numerically.
The optimization problem for gamma under \cref{assump_main:states,assump_main:randomness,assump_main:measurements} of the model takes the form
\begin{equation}
\gamma = - \min_{\tilde{p}_x} \min_{\bar{W}} \min_{\rho \in \Sc} \Tr[\bar W \rho],
\end{equation}
where the minimization over~$\tilde{p}_x$ is constrained by \cref{assump_main:randomness} and the minimization over~$\bar{W}$ is constraint by \cref{assump_main:measurements}.
This optimization is not always easy to carry out.
For example, optimizing over~$\Sc$ is hard in general (e.g., when~$\Sc$ is the set of separable states), although in low dimensions one can often resort to using PPT relaxations~\cite{Peres1996,Horodecki1996}.
Furthermore, the set of feasible~$\bar{W}$ is not convex, and the objective function~$\Tr[\tilde{W \rho}]$ is bilinear in the variables.
For non-convex problems, general nonlinear optimization methods usually only find local optima which may result in witness corrections~$\gamma$ that are smaller than justified.

\subsubsection{Using Hoeffding-Azuma instead of Bentkus} \label{subsubsec:hoeffding}
In this work, we have chosen to use Bentkus' inequality to bound the relevant tail probabilities in \cref{thm:pvalue,thm:CI}, because this is often slightly tighter than the more common Hoeffding-Azuma inequality. We have stated Bentkus' inequality in \cref{lemma:bentkus} in \cref{app:proof2}. Whenever this lemma applies, the Hoeffding-Azuma inequality also applies. So if $R_0,\dots,R_n$ is a supermartingale with $R_0 = 0$ and $-\beta \leq R_i - R_{i-1} \leq 1 - \beta$, then for all $x\geq 0$ Hoeffding-Azuma states that $\Pr[R_n \geq x] \leq e^{-\frac{2x^2}{n}}$. By replacing this in the right places in the proofs of \cref{app:proof2,app:CI}, we can modify the statements of \cref{thm:pvalue,thm:CI}. In particular, \cref{eq:pvalue} can be modified to
\begin{equation}
\Pr[T_n \geq t_n | \Hc_0] \leq e^{-2\frac{\max(t_n - n \beta, 0)^2}{n}},
\end{equation}
with~$\beta$ still defined as in \cref{eq:alpha}, and \cref{eq:vareps_main} can be modified to
\begin{equation}
\alpha = e^{- \frac{n}{2} (\frac{\varepsilon - \gamma}{\ds})^2}.
\end{equation}
Especially this latter modification can be useful, since this makes it easier to compute the confidence interval in \cref{thm:CI}. This is because the above equation can be explicitly solved for~$\varepsilon$ to find
\begin{equation}
\varepsilon = \gamma + \ds\sqrt{\frac{2}{n} \ln(\alpha^{-1})}.
\end{equation}
One may choose to use these equation instead of the theorems as stated, but generally this could come at a slight cost in the tightness of the bounds.

\subsection{General remarks regarding hypothesis testing and the \texorpdfstring{$p$}{p}-value} \label{subsec:pvalue_discussion}
Our rejection analysis of entanglement witness experiments fits into the general framework for hypothesis testing in statistics.
As always, it is important that all relevant parameters are fixed prior to the experiment being carried out.
In our case, this concerns the quantities defined in steps~1.~and~2.~of \cref{subsec:experiment}.
Most notably, this includes the number of trials~$n$ and the significance level~$\alpha$, but also the witness operator, its decomposition, and all other parameters of the~Model Assumptions.
This also implies that one may not simply combine data from multiple experiments (e.g., extending experiments by further trials until a desirable outcome is achieved).
Instead,~$p$-values of different experiments can be combined using known statistical methods~\cite{Elkouss2016,Loughin2004}, or one can carry out a larger, independent experiment instead.

In the framework of hypothesis testing, the~$p$-value is defined as the probability of observing a test statistic~$T$ (in our case, the total normalized score~$T_n$) under the null hypothesis (in our case \cref{assump_main:states}) that is at least as large as the observed value~$t$:
\begin{equation}
p = \Pr[T \geq t | \Hc_0].
\end{equation}
If the~$p$-value is smaller than a previously chosen significance level~$\alpha$, then the null hypothesis is considered to be statistically unlikely to explain the observed~$t$, and we reject the null hypothesis at significance level~$\alpha$.
This means that the~$p$-value is a statement about the observed data in relation to a hypothesis about the distribution of~$T$.
In particular, it does not give any information about the distribution of~$T$ if the null hypothesis does \emph{not} hold.
Neither should the~$p$-value be interpreted as `the probability that the null hypothesis was valid' or as `when the experiment is repeated many times, the null hypothesis is rejected in a~$1-p$ fraction of experiments'.
The~$p$-value is only quantifies the likelihood of a particular assumption about the distribution of~$T$, namely that it is governed by the null hypothesis.
Therefore, it is this a hypothesis that we may plausibly reject upon observing a small~$p$-value.
Finally, we emphasize that the~$p$-value itself can be seen as a random variable, as it is a function of the realization of a randomly observed test statistic~$t$ (whether or not governed by the null hypothesis).
Each experiment is simply a realization of a particular~$p$-value.

\section{Conclusions}
In this work, we proposed two new methods to analyze witness experiments. Both methods are applicable to any source of quantum states that produces a sequence of states on demand.
These states can be arbitrary distributed and correlated -- unlike previous works, which often implicitly assume that the noise is well-behaved so that the central limit theorem applies.
With our rejection analysis, the experimenter can rigorously certify that a device has the capability of producing entangled quantum states.
The statistical confidence is quantified by a bound on the~$p$-value, the probability that an experiment yields data as extreme as the observed data under the assumption that the source produces only separable states.
Hence a small~$p$-value is directly interpreted as strong evidence for the production of entangled states.
The rejection method applies more generally to any witness experiment that is defined by a set of states and a corresponding witness operator.
In particular, it can be used to witness genuine multipartite entanglement.
With our estimation analysis, the experimenter can estimate the average witness value and construct a statistically rigorous confidence interval around this estimate. This method allows the experimenter to conclude whether entanglement was produced on average. This confidence is valid regardless of the type of noise present. The estimation method applies to any general estimation problem. This means that it is not necessary that~$W$ is a witness in the sense of \cref{eq:witness_definition}. Fidelity estimation is therefore also covered by this method.

Both methods we derived are simple to use.
We provide a step-by-step recipe to choose all relevant parameters, collect the necessary experimental data and compute both a bound on the~$p$-value and a average witness estimate with confidence intervals from this collected data.
Our method requires no modification compared to prior methods of performing witness experiments, except possibly the requirement of measuring different settings in random order instead of a fixed, predetermined order.
This requirement is however inevitable if one wishes to deal with arbitrarily correlated noise.
Both of our methods yields a figure of merit -- a bound on the~$p$-value or a point estimate with a confidence interval -- that has a concrete interpretation and is comparable between experiments.
We illustrated our methods with simulations of an experiment in Nitrogen Vacancy centers.

In this work we chose to treat the witness operator as a fixed object that has been chosen with a state to be witnessed in mind.
However, other methods exist in which the witness is modified adaptively based on the measurement data collected so far~\cite{Dai2014}.
The objective of adaptive witnessing is to tune the witness operator (perhaps within a parameterized family) to have maximal detection efficiency, based on the partial knowledge of the noise obtained from prior measurement data.
In particular, one can attempt to identify and counteract coherent components of the noise.
Indeed, coherent and local noise preserve entanglement, and adaptive protocols have been shown to enhance entanglement detection~\cite{Zhou2020}.
An interesting open question is whether the statistical methods we develop in this work can be combined with adaptive witness methods.
This would open ways to increase the detection efficiency of the witness method, while maintaining the robustness against correlated noise.

\begin{acknowledgments}
The authors thank J\'er\'emy Ribeiro for very informative discussions on concentration inequalities and martingale sequences and Liangzhong Ruan for inspiring discussion on estimation problems with non-iid states.
The authors would also like to thank Sophie Hermans and Simon Baier for helpful feedback on earlier versions of this manuscript, and Jonas Helsen for proofreading the manuscript.

BD, MW, and SW acknowledge the NWA Startimpuls.
BD and SW are supported by the Netherlands Organisation for Scientific Research (NWO) through a VIDI grant and by the European Research Council (ERC) trough an ERC Starting Grant QINTERNET.
BD, SW and RH are also supported by the NWO Zwaartekracht Grant Quantum Software Consortium.
RH furthermore acknowledges financial support from NWO through a VICI grant and from the ERC through an ERC Consolidator Grant.
MW is supported by an NWO VENI grant (no.~680-47-459).
MP acknowledges the Marie Sk\l{}odowska-Curie Action `Nanoscale solid-state spin systems in emerging quantum technologies' Spin-NANO (grant agreement no.~676108).
Finally, BD, MP, RH and SW acknowledge the EU Flagship on Quantum Technologies and the Quantum Internet Alliance (funding from EU Horizon 2020 research and innovation programme, Grant No. 820445).
\end{acknowledgments}

\bibliography{NV_project}

\onecolumngrid
\appendix
\section{Game formulation}\label{app:game}
In this first appendix, we expand on \cref{subsec:preliminaries,subsec:measurement-settings}.
We will redefine all objects introduced in the main text, while being slightly more formal and precise.
We will also show that \cref{eq:relation_W_to_score} holds, which relates the definition of the score function to the witness decomposition.
Let us start of with recapturing the notion of a \emph{witness operator} (or simply \emph{witness}).
A witness is defined as a Hermitian observable~$W$ that, for some convex subset of states~$\Sc$, satisfies
\begin{equation}\label{eq:witness_definition_appendix}
\Tr[\rho W] \geq 0 \qquad \forall \rho \in \Sc.
\end{equation}
As described in the main text [\cref{eq:witness-decomposition-observables}], we assume the witness admits a decomposition into the sum of local observables:
\begin{align}
W &= cI + \sum_\xi w_\xi O_\xi^{(1)} \otimes \cdots \otimes O_\xi^{(m)}.   \label{eq:witness-decomposition-app}
\end{align}
Often, it is possible to compute the expectation value of multiple observables $\{O_\xi^{(1)} \otimes \cdots \otimes O_\xi^{(m)}\}_\xi$ from a single measurement setting $M_x^{(1)} \otimes \cdots \otimes M_x^{(m)}$ (we shall label measurement settings by a label~$x$).
To keep track of which observables (labeled by~$\xi$) are related to which measurement setting (labeled by~$x$), we define~$f(\xi) = x$ if the observable~$O_\xi^{(1)} \otimes \cdots \otimes O_\xi^{(m)}$ can be measured by the measurement setting~$M^{(1)}_x \otimes\cdots\otimes M^{(m)}_x$.
This means that there exists a bitstring~$b(\xi) \in \{0,1\}^m$ of length~$m$ such that
\begin{equation}\label{eq:relation-O-M_app}
O_\xi^{(1)} \otimes \cdots \otimes O_\xi^{(m)} = (M_{f(\xi)}^{(1)})^{b_1(\xi)} \otimes \cdots \otimes (M_{f(\xi)}^{(m)})^{b_m(\xi)}.
\end{equation}
To allow for the most general model of measurements, we allow each~$M_x^{(j)}$ in a measurement setting to be measured by a POVM $\{\Pi^{(j),x}_{a}\}_{a \in \Omega_x^{(j)}}$ with outcomes labeled by elements~$a$ in some finite set~$\Omega_x^{(j)}$. That is, we write
\begin{equation}\label{eq:setting_decomposition_app}
M_x^{(j)} = \sum_{a \in \Omega_x^{(j)}} a \, \Pi^{(j),x}_{a}.
\end{equation}
The tensor product POVM for setting $x$ has outcomes $\mathbf{a} = (a_1,...,a_m)$ in $\Omega_x := \Omega_x^{(1)} \times \dots \times \Omega_x^{(m)}$. We denote it by
\begin{align}\label{eq:ideal POVM tensor}
  \Pi^{x}_{\mathbf{a}} := {\Pi}^{(1),x}_{a_1} \otimes \cdots \otimes {\Pi}^{(m),x}_{a_m}.
\end{align}
Sometimes we also need $\Omega^{(j)} := \bigcup_x \Omega_x^{(j)}$ and $\Omega:=\Omega^{(1)}\times\dots\times\Omega^{(m)}$.
We formally define $\Pi^{(j),x}_a := 0$ for $a\in\Omega^{(j)}\setminus\Omega_x^{(j)}$ and $\Pi^{x}_{\mathbf a} := 0$ for $\mathbf a\in\Omega\setminus\Omega_x$.
This is useful in case we want to sum over~$a \in \Omega^{(j)}$ or $\mathbf{a} \in \Omega$ without knowing~$x$.

Using \cref{eq:setting_decomposition_app}, we can write each local observable~$O_\xi^{(j)}$ as
\begin{align}\label{eq:observable_decomposition_app}
O_\xi^{(j)} = \left( M_{f(\xi)}^{(j)} \right)^{b_j(\xi)} = \Bigl( \sum_{a \in \Omega^{(j)}} a \, \Pi^{(j),{f(\xi)}}_a \Bigr)^{b_j(\xi)}
= \sum_{a \in \Omega^{(j)}} a^{b_j(\xi)} \, \Pi^{(j),{f(\xi)}}_a,
\end{align}
where the last equality follows from the facts that~$b_j(\xi) \in \{0,1\}$ and POVM elements sum to the identity~$I$, and we use the convention that~$a^0=1$,~$Q^0 = I$ for any number~$a$ and operator~$Q$. Hence by plugging \cref{eq:observable_decomposition_app} into \cref{eq:witness-decomposition-app}, the witness operator can be further decomposed as
\begin{align}
W
= cI + \sum_\xi w_\xi  \bigotimes_{j=1}^m \sum_{a \in \Omega^{(j)}} a^{b_j(\xi)} \, \Pi^{(j),{f(\xi)}}_a
= cI + \sum_{\mathbf{a}\in\Omega} \sum_x  \Bigl( \sum_{\xi: f(\xi)=x} w_\xi \prod_{j=1}^m a_j^{b_j(\xi)} \Bigr) \Pi^{x}_{\mathbf{a}}. \label{eq:B6}
\end{align}
The last equality follows from rearrangement and \cref{eq:ideal POVM tensor}.
\Cref{tab:notation} in the main text summarizes all relevant objects.

As discussed in \cref{subsec:preliminaries,subsec:measurement-settings}, the witness experiment now proceeds in each round by selecting a measurement setting~$x$ at random, measuring the corresponding POVMs on each subsystem, and, upon obtaining outcomes~$\mathbf a \in \Omega_x$, assigning a score according to \cref{eq:score_function_xi}:
\begin{align}\label{eq:score_function_xi_appendix}
s(x,\mathbf a) = -\frac1{p_x} \sum_{\xi: f(\xi)=x} w_\xi \prod_{j=1}^m a_j^{b_j(\xi)}.
\end{align}
It is useful to define $s(x,\mathbf a) := 0$ if $\mathbf a \not\in\Omega_x$.
Now it becomes apparent why we had written~$W$ in the form \cref{eq:B6}: it allows us to plug in the definition of the score function \cref{eq:score_function_xi_appendix}. As a consequence, we find that
\begin{equation}\label{eq:relation_W_to_score_app}
W = cI - \sum_{x, \mathbf a} p_x \,  s(x, \mathbf{a}) \, \Pi^{x}_{\mathbf{a}},
\end{equation}
showing that \cref{eq:relation_W_to_score} holds.
Finally, we define the algebraic minimum and maximum score, and their difference, as
\begin{equation}
\label{eq:score_max_min_app}
\smin := \min_{x, \mathbf{a}} s(x, \mathbf{a}), \qquad
\smax := \max_{x, \mathbf{a}} s(x, \mathbf{a}), \qquad
\ds := \smax - \smin.
\end{equation}

\section{Formal definition of the Model Assumptions and Null Hypothesis}\label{app:model}
\begin{table}
	\caption{Summary of the random variables present in the model and analysis}
	\ra{1.5}
	\centering
	\begin{tabularx}{\linewidth}{L{0.3} L{0.2}  L{0.2} L{0.2}  L{0.4}}
		\toprule
		Object & Random variable & Type & Restrictions \\
		\midrule
		History &~$H_i$ & arbitrary & \Cref{assumption:sequential} \\
		States &~$\rho_i$ & arbitrary & \Cref{assumption:states} \\
		Measurement setting &~$X_i$ & discrete & \Cref{assumption:randomness} \\
		POVM elements &~$\tilde{\Pi}^{(j)}_{i,a}$ & arbitrary & \Cref{assumption:measurements} \\
		Measurement outcomes &~$\mathbf{A}_i$ & discrete & Born's rule [\cref{eq:born again}] \\
		\midrule
		Score &~$S_i$ & discrete & \Cref{eq:score_random_variable} \\
		Total normalized score &~$T_i$ & discrete & \Cref{eq:app total normalized score} \\
		\bottomrule
	\end{tabularx}
	\label{tab:RV}
\end{table}

In this appendix, we expand on the model as formulated in \cref{subsec:model_informal} and formalize further the Model Assumptions and Null Hypothesis that enter our description of the experiment. This allows us to rigorously draw conclusions from the witness experiment.
The \emph{Model Assumptions} are the collection of the following three assumptions:

\begin{enumerate}[label=(\Roman*),leftmargin=*]
\item \textbf{Sequentiality.}\label[assump]{assumption:sequential}
The experiment is performed in rounds~$i=1,2,\dots,n$.
For each round, there are random variables that model the state~$\rho_i$, measurement setting~$X_i$, POVM~$\{ \tilde \Pi^{(j)}_{i,a} \}_a$ and measurement outcome~$\mathbf A_i$ (defined in the remaining assumptions and summarized in \cref{tab:RV}).
Apart from the restrictions implied by the remaining assumptions, all these random variables in round~$i$ can depend arbitrarily on all previous rounds.
We formalize this by random variables~$H_i$ that models the history of the experiment up until, but excluding, round~$i$.
We assume that~$H_i$ contains at least
\begin{enumerate}
\item\label[assump]{assumption:inclusion}
the history of the preceding round, i.e, the~$\sigma$-algebra of~$H_{i-1}$ is contained in the~$\sigma$-algebra of~$H_i$:
\begin{align}
	\sigma(H_{i-1}) \subseteq \sigma(H_i) \qquad \forall i=2,\dots,n.
\end{align}
\item\label[assump]{assumption:history}
the state, measurement setting and outcome of round $i-1$, i.e.,
\begin{equation}
	\sigma(X_{i-1}, \mathbf A_{i-1}, \rho_{i-1}) \subseteq \sigma(H_i)
	 \qquad \forall i=2,\dots,n,
\end{equation}
By \cref{assumption:inclusion}, this implies that $X_1,\mathbf{A}_1,\rho_1,...,X_{i-1},\mathbf{A}_{i-1},\rho_{i-1}$ are all contained in the history~$H_i$.
\end{enumerate}

\item \textbf{Trusted randomness.}\label[assump]{assumption:randomness}
The measurement setting for the~$i$-th round of the experiment is modeled by a random variable~$X_i$.
We assume that the distribution of $X_i$ given the history~$H_i$ and the state~$\rho_i$ (we denoted this informally as~$\tilde p_{i,x}$ in the main text) is close to a known distribution~$p_x$, i.e., there is~$\tau>0$ such that
\begin{align}
  \left\lvert \Pr[X_i=x | H_i,\rho_i] - p_x \right\rvert \leq \tau \qquad \forall i, x
\end{align}
holds almost surely.
We further assume that
\begin{align}\label{eq:tau bound}
  \tau < \min_x p_x.
\end{align}

\item \textbf{Trusted measurements.}\label[assump]{assumption:measurements}
The measurement in the~$i$-th round of the experiment is modeled by random POVMs $\{ \tilde\Pi^{(j)}_{i,a} \}_a$ that are close to known POVMs~$\{ \Pi^{(j),x}_a \}_{a\in\Omega^{(j)}_x}$, where we recall that $x$~labels the measurement setting, $j=1,\dots,m$ the subsystem, and $\Omega^{(j)}_x$ is the set of possible outcomes
(we formally set $\Pi^{(j),x}_a := 0$ for $a \not\in\Omega^{(j)}_x$).
We model this by assuming that for each of these POVMs $\{ \tilde\Pi^{(j)}_{i,a} \}_a$ there exists a constant~$\delta_j>0$ such that for each~$j$, almost surely,
\begin{align}\label{eq:trusted meas app}
  \bigl\lVert \tilde\Pi^{(j)}_{i,a} - \Pi^{(j),X_i}_a \bigr\rVert_\infty \leq \delta_j \qquad \forall i, j, a.
\end{align}
We further assume that
\begin{align}\label{eq:same outcomes}
  \tilde\Pi^{(j)}_{i,a} = 0 \qquad \forall a \not\in \Omega^{(j)}_{X_i} \qquad \forall i, j, a,
\end{align}
which means the noisy measurements have the same sets of possible outcome as the ideal measurements.
Finally, we assume that the measurement outcomes~$\mathbf A_i$ follow Born's rule, so we demand that
\begin{align}\label{eq:born again}
	\Pr\bigl[\mathbf A_i=\mathbf a\big | H_i,\rho_i,X_i,\{ \tilde\Pi^{(j)}_{i,a} \} \bigr]
= \Tr\bigl[	\rho_i \bigl( \tilde\Pi^{(1)}_{i,a_1} \ot \cdots \ot \tilde\Pi^{(m)}_{i,a_m} \bigr) \bigr]
 \qquad \forall i, \mathbf a.
\end{align}
\end{enumerate}

\noindent
In terms of the above data, the score of the~$i$-th round is given by the following random variable:
\begin{align}\label{eq:score_random_variable}
S_i := s(X_i, \mathbf A_i),
\end{align}
where~$s(x, \mathbf{a})$ is the score function as defined in \cref{eq:score_function_xi_appendix}.

The Model \Cref{assumption:sequential,assumption:measurements,assumption:randomness} are sufficient for analyzing a witness estimation experiment. To analyze a witness rejection experiment, we will also need to define the \emph{Null Hypothesis}. We formally define this as follows (an informal definition was already given in \cref{subsec:model_informal}):

\begin{enumerate}[label=($\Hc_0$), leftmargin=50pt, topsep=12pt]
	\item \textbf{Null Hypothesis: states in~$\Sc$.}\label[hypothesis]{assumption:states}
	The quantum states $\rho_i$ 
  almost surely take values in the set~$\Sc$.
\end{enumerate}
\noindent The Model Assumptions together with this Null Hypothesis are sufficient for performing the witness rejection experiment, with the objective of rejecting \cref{assumption:states}.

\section{Proof of Theorem~\ref{thm:gamma} -- correcting for imperfect randomness and measurements}\label{app:theorem1}
In this appendix, we discuss how the witness inequality must be corrected in the case of imperfect measurements (noisy POVM elements) and imperfect randomness.
To do so, we first introduce shorthand notation for the noisy POVM elements analogously to \cref{eq:ideal POVM tensor}:
\begin{align}
  \tilde\Pi_{i,\mathbf{a}} &:= \tilde{\Pi}^{(1)}_{i,a_1} \otimes \cdots \otimes \tilde{\Pi}^{(m)}_{i,a_m}
\end{align}
Moreover, we define the \emph{effectively implemented operator} $\bar W_i$ for each round~$i$ as [cf.~\cref{eq:relation_Wnoise_to_score,eq:relation_W_to_score_app}]
\begin{equation}\label{eq:W eff app}
  \bar W_i
= cI - \sum_{x, \mathbf a} \Pr[X_i=x | H_i,\rho_i]  s(x,\mathbf a) \, \bar\Pi^x_{i,\mathbf a},
\end{equation}
in terms of the following random variables
\begin{equation}\label{eq:bar pi x app}
  \bar\Pi^x_{i,\mathbf a} := \frac {\Ebb\bigl[ \tilde\Pi_{i,\mathbf a} \mathbf 1_{[X_i=x]} | H_i, \rho_i \bigr]} {\Pr[X_i=x | H_i, \rho_i]}.
\end{equation}
We note that the denominator is nonzero almost surely by \cref{eq:tau bound}, so that $\{ \bar\Pi^x_{i,\mathbf a} \}$ is almost surely well-defined and a POVM for each~$i$ and $x$.
These POVM elements can be interpreted as the expected implemented POVM elements in round~$i$, conditioned on the event that setting $X_i = x$ is realized.
We now bound the deviation of the effectively implemented operator $\bar{W}_i$ from the target operator~$W$ as a result of imperfect measurements and randomness.

\begin{reptheorem}{thm:gamma}
  Let~$W$ be a Hermitian operator of the form of \cref{eq:B6} [not necessarily a witness in the sense of \cref{eq:witness_definition_appendix}].
  Suppose the experiment is modeled by the Model \cref{assumption:randomness,assumption:measurements,assumption:sequential} in \cref{app:model}.
  Define the score function as in \cref{eq:score_function_xi_appendix}.
  Denote by $\bar{W}_i$ the effectively implemented operators as in \cref{eq:W eff app}.
  Then, in every round~$i$,
  \begin{align}\label{eq:thm1_statement_app}
    \lVert \bar W_i - W \rVert_\infty \leq \gamma
  \end{align}
  holds almost surely, where the \emph{witness correction}
  \begin{align}\label{eq:app gamma}
    \gamma := \gamma_1 + \gamma_2
  \end{align}
  is the sum of the \emph{random number generation correction} and the \emph{measurement correction} defined by
  \begin{align}
    \gamma_1 := \tau \sum_x  \max_{\mathbf a}\;\lvert s(x,\mathbf a) \rvert,
    \qquad \mbox{and} \qquad
    \gamma_2 := \sum_\xi \lvert w_\xi \rvert \gamma_2(\xi),
  \end{align}
  respectively, in terms of
  \begin{equation}\label{eq:gamma shorthand}
  \gamma_2(\xi) := \sum_{j=1}^m  \Big(\prod_{k=1}^{j-1} ( \lambda_\xi^{(k)} + \epsilon^{(k)}_\xi ) \Big) \epsilon^{(j)}_\xi \Big(\prod_{k=j+1}^{m} \lambda_\xi^{(k)}\Big), \qquad \epsilon^{(j)}_\xi := b_j(\xi) \delta_j \sum_{a \in \Omega^{(j)}_{f(\xi)}} \lvert a \rvert
  \qquad\text{and}\qquad
  \lambda_\xi^{(j)} := \lVert O_\xi^{(j)} \rVert_\infty.
  \end{equation}
\end{reptheorem}
\begin{proof}
Define
\begin{align}\label{eq:W intermediate app}
  \check W_i := cI - \sum_x p_x \sum_{\mathbf a} s(x,\mathbf a) \, \bar\Pi^x_{i,\mathbf a}.
\end{align}
We will show that, almost surely, $\lVert \check W_i - \bar W_i \rVert_\infty \leq \gamma_1$ and $\lVert W - \check W_i \rVert_\infty \leq \gamma_2$, which together imply the theorem by the triangle inequality.
For the former, we find that by comparing the definitions of $\bar{W}_i$ and $\check{W}_i$, \cref{eq:W eff app,eq:W intermediate app} respectively,  that
\begin{align}
  \lVert \check W_i - \bar W_i \rVert_\infty
&= \lVert \sum_x \left( \Pr[X_i=x | H_i,\rho_i] - p_x \right) \sum_{\mathbf a} s(x,\mathbf a) \bar\Pi^x_{i,\mathbf a} \rVert_\infty \\
&\leq \tau \sum_x \lVert \sum_{\mathbf a} s(x,\mathbf a) \bar\Pi^x_{i,\mathbf a} \rVert_\infty \\
&\leq \tau \sum_x \max_{\mathbf a}\; \lvert s(x,\mathbf a) \rvert = \gamma_1,
\end{align}
using \cref{assumption:randomness} in the first inequality and
using in the second inequality that the operators $\{\bar\Pi^x_{i,\mathbf a}\}$ form a POVM for each~$i$ and~$x$.

To show that $\lVert W - \check W_i \rVert_\infty \leq \gamma_2$, we start by comparing their respective definitions~\cref{eq:W intermediate app,eq:relation_W_to_score_app}, which gives
\begin{align}
  W - \check W_i
= \sum_x p_x \sum_{\mathbf a} s(x,\mathbf a) \left( \bar\Pi^x_{i,\mathbf a} - \Pi^x_{\mathbf a} \right)
= \sum_\xi w_\xi \sum_{\mathbf a} \bigl( \prod_{j=1}^m a_j^{b_j(\xi)} \bigr) \left( \Pi^{f(\xi)}_{\mathbf a} - \bar\Pi^{f(\xi)}_{i,\mathbf a} \right),
\end{align}
where we inserted the definition of the score function~\cref{eq:score_function_xi_appendix} in the second step.
Thus,
\begin{align}
  \lVert W - \check W_i \rVert_\infty
\leq \sum_\xi \lvert w_\xi \rvert \, \lVert \sum_{\mathbf a} \bigl( \prod_{j=1}^m a_j^{b_j(\xi)} \bigr) \left( \Pi^{f(\xi)}_{\mathbf a} - \bar\Pi^{f(\xi)}_{i,\mathbf a} \right) \rVert_\infty,
\end{align}
so it remains to prove that
\begin{align}\label{eq:gamma 2 xi target}
  \lVert \sum_{\mathbf a} \bigl( \prod_{j=1}^m a_j^{b_j(\xi)} \bigr) \left( \Pi^{f(\xi)}_{\mathbf a} - \bar\Pi^{f(\xi)}_{i,\mathbf a} \right) \rVert_\infty
\leq \gamma_2(\xi),
\end{align}
for all~$\xi$ to conclude the proof.
Using the definition of~$\bar\Pi^{f(\xi)}_{i,\mathbf a}$, \cref{eq:bar pi x app}, we find that
\begin{align}
\sum_{\mathbf a} \bigl( \prod_{j=1}^m a_j^{b_j(\xi)} \bigr) \left( \Pi^{f(\xi)}_{\mathbf a} - \bar\Pi^{f(\xi)}_{i,\mathbf a} \right)
&= \sum_{\mathbf a} \bigl( \prod_{j=1}^m a_j^{b_j(\xi)} \bigr) \left( \Pi^{f(\xi)}_{\mathbf a} - \frac {\Ebb\bigl[ \tilde\Pi_{i,\mathbf a} \mathbf 1_{[X_i=f(\xi)]} | H_i, \rho_i \bigr]} {\Pr[X_i=f(\xi) | H_i, \rho_i]} \right) \\
&= \sum_{\mathbf a} \bigl( \prod_{j=1}^m a_j^{b_j(\xi)} \bigr) \frac {\Ebb\bigl[ (\Pi^{f(\xi)}_{\mathbf a} - \tilde\Pi_{i,\mathbf a}) \mathbf 1_{[X_i=f(\xi)]} | H_i, \rho_i \bigr]} {\Pr[X_i=f(\xi) | H_i, \rho_i]} \\
&= \frac {\Ebb\bigl[ \sum_{\mathbf a} ( \prod_{j=1}^m a_j^{b_j(\xi)} )  (\Pi^{f(\xi)}_{\mathbf a} - \tilde\Pi_{i,\mathbf a}) \mathbf 1_{[X_i=f(\xi)]} | H_i, \rho_i \bigr]} {\Pr[X_i=f(\xi) | H_i, \rho_i]}.
\end{align}
Therefore, we find that
\begin{align}\label{eq:xi norm bound app}
  \lVert \sum_{\mathbf a} \bigl( \prod_{j=1}^m a_j^{b_j(\xi)} \bigr) \left( \Pi^{f(\xi)}_{\mathbf a} - \bar\Pi^{f(\xi)}_{i,\mathbf a} \right) \rVert_\infty
\leq \frac {\Ebb\bigl[ \lVert \sum_{\mathbf a} \bigl( \prod_{j=1}^m a_j^{b_j(\xi)} \bigr)  (\Pi^{f(\xi)}_{\mathbf a} - \tilde\Pi_{i,\mathbf a}) \rVert_\infty \mathbf 1_{[X_i=f(\xi)]} | H_i, \rho_i \bigr]} {\Pr[X_i=f(\xi) | H_i, \rho_i]}.
\end{align}
To established the desired \cref{eq:gamma 2 xi target}, it therefore remains to show that, almost surely,
\begin{align}\label{eq:last eq standing}
\Big\lVert \sum_{\mathbf a} \bigl( \prod_{j=1}^m a_j^{b_j(\xi)} \bigr) (\Pi^{f(\xi)}_{\mathbf a} - \tilde\Pi_{i,\mathbf a}) \Big\rVert_\infty \mathbf 1_{[X_i=f(\xi)]}
\;\leq\; \gamma_2(\xi) \, \mathbf 1_{[X_i=f(\xi)]}.
\end{align}
We will show \cref{eq:last eq standing} in the remainder of the proof.
We expand~$\tilde\Pi_{i,\mathbf{a}} - \Pi^{x}_{\mathbf{a}}$ using a telescoping sum of the form
\begin{equation}
C_1 C_2 C_3 - B_1 B_2 B_3 = (C_1 - B_1) B_2 B_3 + C_1 (C_2 - B_2) B_3 + C_1 C_2 (C_3 - B_3).
\end{equation}
In this case, with the tensor product and~$m$ terms, we get
\begin{equation}
  \tilde\Pi_{i,\mathbf{a}} - \Pi^{x}_{\mathbf{a}}
= \sum_{j=1}^m \Bigl( \bigotimes_{k=1}^{j-1} \tilde{\Pi}^{(k)}_{i,a_k} \Bigr) \otimes \left( \tilde{\Pi}^{(j)}_{i,a_j} - \Pi^{(j),x}_{a_j} \right) \otimes \Bigl( \bigotimes_{k=j+1}^{m} \Pi^{(k),x}_{a_k} \Bigr).
\end{equation}
Using this, and by distributing the product and sum in the below expression over the tensor factors, we find that
\begin{align}
\sum_{\mathbf a} \Bigl( \prod_{j=1}^m a_j^{b_j(\xi)} \Bigr) \Bigl( \tilde{\Pi}_{i,\mathbf a} - \Pi^{f(\xi)}_{\mathbf{a}} \Bigr)
= \sum_{j=1}^m \Bigl( \bigotimes_{k=1}^{j-1} \tilde O^{(k)}_{i,\xi} \Bigr) \otimes \left( \tilde O^{(j)}_{i,\xi} - O^{(j)}_\xi \right) \otimes \Big( \bigotimes_{k=j+1}^{m} O^{(k)}_\xi \Big),
\end{align}
where
\begin{align}
\tilde O^{(j)}_{i,\xi} &:= \Bigl( \sum_{a \in \Omega^{(j)}} a \, \tilde{\Pi}^{(j)}_{i,a} \Bigr)^{b_j(\xi)}
=\sum_{a \in \Omega^{(j)}} a^{b_j(\xi)} \tilde{\Pi}^{(j)}_{i,a} ,
\end{align}
by the fact that~$b_j(\xi) \in \{0,1\}$ and POVM elements sum to the identity, in analogy to their ideal counterparts~$O_\xi^{(j)}$ given in \cref{eq:observable_decomposition_app}.
Therefore, using the fact that the operator norm~$\lVert\cdot\rVert_\infty$ is multiplicative with respect to tensor products (i.e.~$\lVert C \otimes B\rVert_\infty = \lVert C\rVert_\infty \lVert B\rVert_\infty$), we find that
\begin{align}
\Big\lVert \sum_{\mathbf{a}} \Big(\prod_{j=1}^{m} a_j^{b_j(\xi)}\Big) (\tilde{\Pi}_{i,\mathbf{a}} - \Pi^{f(\xi)}_{\mathbf{a}} ) \Big\rVert_\infty
&\leq \sum_{j=1}^m  \Big(\prod_{k=1}^{j-1} \Big\lVert \tilde{O}^{(k)}_{i,\xi}  \Big\rVert_\infty\Big)  \Big\lVert \tilde{O}^{(j)}_{i,\xi} - {O}^{(j)}_\xi \Big\rVert_\infty   \Big(\prod_{k=j+1}^{m} \Big\lVert {O}^{(k)}_\xi  \Big\rVert_\infty\Big). \label{eq:B38}
\end{align}
We continue by bounding~$\lVert \tilde{O}^{(j)}_{i,\xi} - {O}^{(j)}_\xi \rVert_\infty$ and~$\lVert \tilde{O}^{(j)}_{i,\xi} \rVert_\infty$.
We start with bounding the first, by considering the two cases~$b_j(\xi) = 0$ and~$b_j(\xi)=1$ separately.
If~$b_j(\xi) = 0$, then
\begin{equation}
  \tilde{O}_{i,\xi}^{(j)} - O_{\xi}^{(j)}
= \sum_{a \in \Omega^{(j)}} a^{b_j(\xi)} \left( \tilde{\Pi} ^{(j)}_{i,a}-{\Pi}^{(j),{f(\xi)}}_{a} \right)
= \sum_{a \in \Omega^{(j)}} \left( \tilde{\Pi}^{(j)}_{i,a}-{\Pi}^{(j),{f(\xi)}}_{a} \right)
= I - I = 0,
\end{equation}
so then~$\lVert \tilde{O}_{i,\xi}^{(j)} - O_\xi^{(j)} \rVert_\infty = 0$.
On the other hand, if~$b_j(\xi) = 1$, then, almost surely,
\begin{align}
\Big\lVert \tilde O_{i,\xi}^{(j)} - O_\xi^{(j)} \Big\rVert_\infty \mathbf 1_{[X_i=f(\xi)]}
&=  \Big\lVert  \sum_{a \in \Omega^{(j)}} a (\tilde{\Pi}^{(j)}_{i,a}-{\Pi}^{(j),{f(\xi)}}_{a}) \Big\rVert_\infty \mathbf 1_{[X_i=f(\xi)]} \\
& \leq \sum_{a \in \Omega^{(j)}} \lvert a \rvert \Big\lVert \tilde{\Pi}^{(j)}_{i,a}-{\Pi}^{(j),{f(\xi)}}_{a} \Big\rVert_\infty \mathbf 1_{[X_i=f(\xi)]} \\
& = \sum_{a \in \Omega^{(j)}_{f(\xi)}} \lvert a \rvert \Big\lVert \tilde{\Pi}^{(j)}_{i,a}-{\Pi}^{(j),X_i}_{a} \Big\rVert_\infty \mathbf 1_{[X_i=f(\xi)]} \\
&\leq \delta_j \sum_{a \in \Omega^{(j)}_{f(\xi)}} \lvert a\rvert \, \mathbf 1_{[X_i=f(\xi)]}. \label{eq:C24}
\end{align}
For the equality, note that we can assume that $X_i = f(\xi)$ by the indicator function; since then both POVMs have outcomes in~$\Omega^{(j)}_{f(\xi)}$ we can also restrict the summation [cf.~\cref{eq:same outcomes}].
The last step holds by \cref{eq:trusted meas app} in \cref{assumption:measurements}.
Both cases are neatly summarized by
\begin{equation}\label{eq:B32}
\left\lVert \tilde{O}_{i,\xi}^{(j)} - O_\xi^{(j)} \right\rVert_\infty \mathbf 1_{[X_i=f(\xi)]}
\leq \epsilon_\xi^{(j)} \mathbf 1_{[X_i=f(\xi)]},
\end{equation}
with~$\epsilon_\xi^{(j)}$ defined as in \cref{eq:gamma shorthand}.
By the triangle inequality, this in turn implies that
\begin{equation}
\left\lVert \tilde{O}_{i,\xi}^{(j)} \right\rVert_\infty \mathbf 1_{[X_i=f(\xi)]}
\leq \left\lVert O_\xi^{(j)} \right\rVert_\infty \mathbf 1_{[X_i=f(\xi)]} + \left\lVert \tilde{O}_{i,\xi}^{(j)} - {O}_\xi^{(j)} \right\rVert_\infty \mathbf 1_{[X_i=f(\xi)]}
\leq \left( \lambda_\xi^{(j)} + \epsilon^{(j)}_\xi \right) \mathbf 1_{[X_i=f(\xi)]}
\label{eq:B33}
\end{equation}
with~$\lambda_\xi^{(j)} = \lVert O_\xi^{(j)} \rVert_\infty$ defined as in \cref{eq:gamma shorthand}.
Thus, plugging \cref{eq:B33,eq:B32} into \eqref{eq:B38}, we obtain the desired bound \cref{eq:last eq standing}, completing the proof.
\end{proof}

\noindent
As a consequence of this theorem, we obtain two corollaries that related the expected score $\Ebb[S_i|H_i, \rho_i]$ to the witness value $\Tr[W\rho_i]$. These corollaries are used in \cref{thm:pvalue,thm:CI}. We state and proof both below.

\begin{corollary}\label{cor:two-sided}
Let~$W$ be a Hermitian operator of the form of \cref{eq:B6} [not necessarily a witness in the sense of \cref{eq:witness_definition_appendix}].
Suppose that the experiment is modeled by the Model \cref{assumption:randomness,assumption:measurements,assumption:sequential} in \cref{app:model}.
Let the score function be defined as in \cref{eq:score_function_xi_appendix}.
Then, in every round~$i$,
\begin{equation}
\lvert \Tr[W\rho_i] - (c - \Ebb[S_i|H_i, \rho_i]) \rvert \leq \gamma
\end{equation}
holds almost surely, where~$c$ and~$\gamma$ are defined in \cref{eq:B6,eq:app gamma}.
\end{corollary}
\begin{proof}
We first compute
\begin{align}
  \Ebb\bigl[S_i | H_i,\rho_i,X_i,\{ \tilde\Pi^{(j)}_{i,a} \} \bigr]
&= \Ebb\bigl[s(X_i,\mathbf A_i) | H_i,\rho_i,X_i,\{ \tilde\Pi^{(j)}_{i,a} \} \bigr] \\
&= \sum_{\mathbf a} \Ebb\bigl[s(X_i,\mathbf a) \mathbf 1_{[\mathbf A_i=\mathbf a]} | H_i,\rho_i,X_i,\{ \tilde\Pi^{(j)}_{i,a} \} \bigr] \\
&= \sum_{\mathbf a} s(X_i,\mathbf a) \Ebb\bigl[\mathbf 1_{[\mathbf A_i=\mathbf a]} | H_i,\rho_i,X_i,\{ \tilde\Pi^{(j)}_{i,a} \} \bigr] \\
&= \sum_{\mathbf a} s(X_i,\mathbf a) \Pr\bigl[\mathbf A_i=\mathbf a \big| H_i,\rho_i,X_i,\{ \tilde\Pi^{(j)}_{i,a} \} \bigr] \\
&= \sum_{\mathbf a} s(X_i,\mathbf a) \Tr[ \rho_i \tilde\Pi_{i,\mathbf a}],
\end{align}
using Born's rule \eqref{eq:born again} in the last equation.
This implies that
\begin{align}
  \Ebb\bigl[S_i | H_i,\rho_i \bigr]
&= \sum_{\mathbf a} \Ebb\bigl[s(X_i,\mathbf a) \Tr[ \rho_i \tilde\Pi_{i,\mathbf a} ] \big| H_i, \rho_i \bigr] \\
&= \sum_{x,\mathbf a} s(x,\mathbf a) \Ebb\bigl[\mathbf 1_{[\mathbf X_i = x]} \Tr[ \rho_i \tilde\Pi_{i,\mathbf a} ] \big| H_i, \rho_i \bigr] \\
&= \Tr\bigl[ \rho_i \sum_{x,\mathbf a} s(x,\mathbf a) \Ebb\bigl[\mathbf 1_{[\mathbf X_i = x]} \tilde\Pi_{i,\mathbf a} \big| H_i, \rho_i \bigr] \bigr] \label{eq:C40} \\
&= c - \Tr[ \bar W_i \rho_i], \label{eq:C41}
\end{align}
where in \cref{eq:C40} we pull can pull $\rho_i$ out of the conditional expectation value since it is conditioned on $\rho_i$ and in \cref{eq:C41} we use the definition of $\bar W_i$ (given in \cref{eq:W eff app,eq:bar pi x app}).
From this, it follows that
\begin{align}
  \lvert \Tr[W\rho_i] - (c - \Ebb[S_i|H_i, \rho_i]) \rvert
= \lvert \Tr[(W - \bar W_i) \rho_i] \rvert
\leq \lVert W - \bar W_i \rVert_\infty \leq \gamma,
\end{align}
where the last inequality follows from the Theorem, \cref{eq:thm1_statement_app}.
\end{proof}

\begin{corollary}\label{cor:one-sided}
Let~$W$ be an operator of the form of \cref{eq:B6} that is a witness for the set~$\Sc$ [i.e.,~$W$ satisfies \cref{eq:witness_definition_appendix}].
Suppose that the experiment is modeled by the Model \Cref{assumption:randomness,assumption:measurements,assumption:sequential} in \cref{app:model}.
Furthermore, assume that the \cref{assumption:states} holds with respect to this~$\Sc$.
Then, almost surely, $\Ebb[S_i | H_i] \leq c + \gamma$ for all rounds~$i=1,\dots,n$, where~$c$ and~$\gamma$ are defined in \cref{eq:B6,eq:app gamma}.
\end{corollary}
\begin{proof}
From \cref{cor:two-sided}, we directly obtain the one-sided inequality
$
\Ebb[S_i | H_i, \rho_i]    \leq   \gamma + c - \Tr[\rho_i W].
$
Now, \cref{assumption:states} and \cref{eq:witness_definition_appendix} imply that~$\Tr[\rho_i W] \geq 0$ holds almost surely, so that~$\Ebb[S_i | H_i, \rho_i]    \leq   \gamma + c - \Tr[\rho_i W] \leq c + \gamma$.
Thus, $\Ebb[S_i | H_i] = \Ebb[\Ebb[S_i | H_i, \rho_i] | H_i] \leq \Ebb[c + \gamma | H_i] = c + \gamma$.
\end{proof}

\section{Proof of Theorem~\ref{thm:pvalue} -- bounding the~\texorpdfstring{$p$}{p}-value}\label{app:proof2}
The goal of this section is to prove \cref{thm:pvalue}. In this theorem, we derive a bound on the~$p$-value.
The main ingredient for the proof is a concentration inequality that bounds the tail probabilities of (super)martingales.
Here, we use a concentration inequality due to Bentkus~\cite{Bentkus2004,Bentkus2006}, which was used in Ref.~\cite{Elkouss2016} for the analysis of Bell violation experiments.
We state Bentkus' inequality in the form of Theorem~1 in Ref.~\cite{Elkouss2016}.

\begin{lemma}[Bentkus' inequality~\cite{Bentkus2004,Bentkus2006}]\label{lemma:bentkus}
Let~$R_0$,~$R_1$, \dots,~$R_n$ be a supermartingale (with respect to an arbitrary filtration) with~$R_0=0$ and differences bounded as
$-\beta_i \leq R_i - R_{i-1} \leq 1-\beta_i$ almost surely for constants~$\beta_i \in [0,1]$.
Let~$\beta := (\beta_1+\dots+\beta_n)/n$.
Then, for all~$x\in\R$,
\begin{align}
  \Pr[R_n \geq x - n\beta] \leq e F^\circ_{n,\beta}(x),
\end{align}
where
\begin{align}\label{eq:Fc}
  F^\circ_{n,\beta}(x) = F_{n,\beta}(\lfloor x\rfloor)^{1 - (x - \lfloor x\rfloor)} F_{n,\beta}(\lfloor x\rfloor+1)^{x - \lfloor x\rfloor}
\end{align}
is the log-linear interpolation of the survival function of a binomial distribution with parameters~$n$ and~$\beta$,
\begin{align}
F_{n,\beta}(k) &= \Pr[X \geq k | X \sim \Binom(n,\beta)] \\
&= \sum_{l = k}^{n} \binom{n}{l} \beta^l (1-\beta)^{n-l}.
\end{align}
In \cref{eq:Fc}, we define~$0^0 := 1$.
\end{lemma}

\noindent
With this main ingredient in hand, we now state and prove a bound on the~$p$-value~$p := \Pr[T_n \geq t_n | \Hc_0 ]$, where
\begin{align}\label{eq:app total normalized score}
  T_n := \sum_{i=1}^n \frac{S_i - \smin}{\ds}
\end{align}
is the total normalized score after~$n$ rounds (see \cref{eq:score_function_xi_appendix,eq:score_random_variable} for the definition of the score).
We recall that~$\smin$ and~$\smax$ are defined in \cref{eq:score_max_min_app} as the minimum and maximum value, respectively, that can be attained by the score function.

\begin{reptheorem}{thm:pvalue}
	Let~$W$ be an operator of the form of \cref{eq:B6} that is a witness for the set~$\Sc$ [i.e., which satisfies \cref{eq:witness_definition_appendix}].
	Assume that the experiment is modeled by the Model \cref{assumption:randomness,assumption:measurements,assumption:sequential} in \cref{app:model} and consider the Null \cref{assumption:states} with respect to~$\Sc$.
	Then, for all~$t_n\in\R$,
	\begin{align}
	\Pr[T_n \geq t_n | \Hc_0] \leq e F^\circ_{n,\beta}(t_n), \qquad\text{where}\qquad \beta := \min\Bigl(1, \frac{c + \gamma - \smin}{\ds} \Bigr),  \label{eq:pvalue with alpha}
	\end{align}
	and where~$c$,~$\gamma$, and~$F^\circ_{n,\beta}$ are defined in \cref{eq:B6,eq:app gamma,eq:Fc}, respectively.
\end{reptheorem}
\begin{proof}
We first verify that~$\beta\geq0$, so that~$F^\circ_{n,\beta}$ is well-defined and we can later apply \cref{lemma:bentkus}.
Now, \cref{assumption:states} and \cref{eq:witness_definition_appendix} imply that~$\Tr[\rho_i W] \geq 0$.
Therefore by the decomposition of $W$ from~\cref{eq:relation_W_to_score_app}, we find that
\begin{align}\label{eq:valid omega bound}
  0
\leq \Tr[W\rho_i]
= c - \sum_{x, \mathbf{a}} p_x \Tr\Bigl[\rho_i \Pi_{\mathbf{a}}^{x} \Bigr] s(x, \mathbf{a})
\leq c - s_{\min}.
\end{align}
Since~$\gamma\geq0$, it follows that indeed~$\beta\geq0$.

Next, we verify that the bound in \cref{eq:pvalue with alpha} holds for~$\beta=1$.
Indeed,~$\Pr[T_n \geq t_n]>0$ only if~$t_n\leq n$, since~$T_n\in[0,n]$.
But in this case,~$F^\circ_{n,1}(t_n) = 1$, so the bound reads~$\Pr[T_n \geq t]\leq e$, which holds trivially.

Now assume that~$\beta\in[0,1)$, so that
\begin{align}\label{eq:happy alpha}
  \beta = \frac{c + \gamma - \smin}{\ds}.
\end{align}
The proof comes down to constructing a suitable supermartingale and applying \cref{lemma:bentkus}.
For~$i=0,1,\dots,n$, define
\begin{equation}\label{eq:B1}
R_i := \sum_{l=1}^{i} \frac{S_l - c  - \gamma}{\ds}.
\end{equation}
By convention of the empty sum, this means~$R_0 = 0$. Note that~$R_i$ is affinely related to~$T_i$, the total normalized score after round~$i$.

We first show that~$R_0$,~$R_1$, \dots,~$R_n$ is a supermartingale with respect to~$H_1$, \dots,~$H_n$, i.e.,
\begin{equation}\label{eq:supermartingale}
  \Ebb[R_i | H_1,\dots,H_i] \leq R_{i-1} \qquad \forall i=1,\dots,n.
\end{equation}
Note that, for any~$i=1,\dots,n$, we have the recursion formula
\begin{align}\label{eq:R recursion}
  R_i = R_{i-1} + \frac{S_i - c - \gamma}{\ds}.
\end{align}
In view of \cref{eq:score_random_variable,eq:B1},~$R_{i-1}$ depends only on the measurement settings and outcomes of the first~$i-1$ rounds.
Thus,
$\Ebb[R_{i-1} | H_1,\dots,H_i] = R_{i-1}$ by \cref{assumption:history}.
Moreover,~$\Ebb[S_i | H_1,\dots, H_i] = \Ebb[S_i|H_i]$ by \cref{assumption:inclusion}.
Therefore,
\begin{align}
  \Ebb[R_i | H_1,\dots,H_i]
= R_{i-1} + \frac{\Ebb[S_i|H_i] - c - \gamma}{\ds}
\leq R_{i-1}.
\end{align}
where the last inequality holds by \cref{cor:one-sided}, which asserts that~$\Ebb[S_i|H_i]  \leq c + \gamma$ almost surely. Thus,~$R_0,R_1,\dots,R_n$ is a supermarginale as claimed.

Next we bound the differences.
Since~$\smin \leq S_i \leq \smax$, we find using \cref{eq:R recursion} that
\begin{align}
\frac{\smin - c - \gamma}{\ds}
\leq R_i - R_{i-1}
= \frac{S_i - c - \gamma}{\ds}
\leq \frac{\smax - c - \gamma}{\ds}.
\end{align}
Thus, it holds that~$-\beta \leq R_i - R_{i-1} \leq 1-\beta$ since~$\beta$ is given by \cref{eq:happy alpha}.
Indeed,
\begin{align}
  -\beta
= \frac{\smin - c - \gamma}{\ds}, \qquad
  1-\beta
= \frac{\smax - c - \gamma}{\ds}.
\end{align}
Thus we can apply \cref{lemma:bentkus}, which gives
\begin{equation}\label{eq:bentkus in proof}
  \Pr[R_n \geq x - n\beta] \leq e F^\circ_{n,\beta}(x) \qquad\forall x\in\R.
\end{equation}
To complete the proof, we need to relate~$T_n$ to~$R_n$ and evaluate at the appropriate value of~$x$.
By comparing \cref{eq:B1,eq:app total normalized score}, we find that
\begin{align}
  R_n
= \sum_{l=1}^n \frac{S_l - \smin - (c + \gamma - \smin)}{\ds}
= T_n - n \beta.
\end{align}
If we combine this with \cref{eq:bentkus in proof}, we arrive at
\begin{align}
  \Pr[T_n \geq t_n | \Hc_0]
= \Pr[R_n + n \beta \geq t_n | \Hc_0]
= \Pr[R_n \geq t_n - n\beta | \Hc_0]
\leq e F^\circ_{n,\beta}(t_n).
\end{align}
This completes the proof.
\end{proof}

\section{Proof of Theorem~\ref{thm:CI} -- confidence intervals for the witness estimate}\label{app:CI}
In this appendix we will prove \cref{thm:CI} in the main text.
This theorem establishes a~$(1-2\alpha)$ two-sided and a~$(1-\alpha)$ one-sided  confidence interval around an estimate~$\hat{w}_n$ of the average witness value~$\braket{W}_n$. Recall [from \cref{eq:witness_estimate_main}] that the witness estimate is computed from the score as
\begin{equation}
\label{eq:witness_estimate}
\hat{W}_n := c - \frac{1}{n} \sum_{i=1}^{n} S_i.
\end{equation}
It is a point estimate of the average realized witness value~$\braket{W}_n$, which is defined as [recall \cref{eq:avg_witness}]
\begin{equation}
\label{eq:witness_def_app}
\braket{W}_n := \frac{1}{n} \sum_{i=1}^n \Tr[\rho_i W].
\end{equation}

We start with a lemma that establishes the some properties of the function~$F^\circ_{n,\half}$ from Bentkus' inequality, so that it has a well-defined inverse.
\begin{lemma}\label{lemma:Fc-props}
	Let~$F^\circ_{n,\half}(x)$ be defined as in \cref{eq:Fc} for any~$n\in\N$ (with~$\beta=\half$). Then for any~$n\in\N$ the function~$x \mapsto F^\circ_{n,\half}(x)$ has the following properties:
	\begin{enumerate}[label=(\alph*)]
		\item~$F^\circ_{n,\half}(x)$ is a strictly decreasing and continuous in~$x$ on the interval~$[\frac{n}{2}, n]$; and
		\item~$F^\circ_{n,\half}([\frac{n}{2}, n]) \supseteq [\frac{1}{2^n}, \half]$.
	\end{enumerate}
	Hence, for any~$y \in [\frac{1}{2^n}, \half]$, there is a unique~$x \in [\frac{n}{2}, n]$ such that~$F^\circ_{n,\half}(x) = y$.
\end{lemma}
\begin{proof}
	Recall from \cref{eq:Fc} that~$F^\circ_{n,\half}$ interpolates the survival function~$F_{n,\half}(k) = \Pr[X \geq k]$ of a binomial random variables~$X \sim \Binom(n,\half)$ at non-integer points by the log-linear function
	\begin{align}\label{eq:log linear restated}
	F^\circ_{n,\half}(x) = F_{n,\half}(\lfloor x\rfloor)^{1 - (x - \lfloor x\rfloor)} F_{n,\half}(\lfloor x\rfloor+1)^{x - \lfloor x\rfloor}.
	\end{align}

	\noindent (a) Since~$F_{n,\half}(k)$ is strictly decreasing in~$k$ and strictly positive for all~$k = 0,\dots,n$, and since the logarithm is strictly monotonic and continuous, it is clear that~$F^\circ_{n,\half}$ is strictly decreasing and continuous in~$x$ for all~$x \in [\frac{n}{2},n]$.

	\noindent (b) By the convention that~$0^0 = 1$ in the log-linear interpolation, it follows that~$F^\circ_{n,\half}(n) = F_{n,\half}(n) = \frac{1}{2^n}$.
	Now we show that~$F^\circ_{n,\half}(\frac{n}{2}) \geq \half$. To do so, we use the symmetry of the binomial distribution with parameter half. We observe that
	\begin{equation}
	F_{n,\half}(k) = \sum_{i=k}^{n} \binom{n}{i} \frac{1}{2^n} = \sum_{i=k}^{n} \binom{n}{n-i} \frac{1}{2^n} = \sum_{i=0}^{n-k} \binom{n}{i} \frac{1}{2^n} = 1 - F_{n,\half}(n-k+1).
	\end{equation}
	So, for even~$n$, we find that
	\begin{align}
	2 F^\circ_{n,\half}(\frac{n}{2}) = 2 F_{n,\half}(\frac{n}{2}) = F_{n,\half}(\frac{n}{2}) + 1 - F_{n,\half}(\frac{n}{2}+1) \geq 1
	\end{align}
	since~$F_{n,\half}(\frac{n}{2}+1) > F_{n,\half}(\frac{n}{2})$. And for odd~$n$ is odd, then we find [by property (a)]
	\begin{align}
	2 F^\circ_{n,\half}(\frac{n}{2}) &\geq 2 F^\circ_{n,\half}(\frac{n+1}{2}) = F_{n,\half}(\frac{n+1}{2}) + 1 - F_{n,\half}(n - \frac{n+1}{2} + 1) = 1.
	\end{align}
	Hence in either case~$F^\circ_{n,\half}(\frac{n}{2}) \geq \half$. By property (a) and~$F^\circ_{n,\half}(n) = \frac{1}{2^n}$ statement (b) follows.
\end{proof}
With this lemma in hand, we can state and proof the theorem. The main ingredient in the proof is again \cref{lemma:bentkus} applied to a suitably chosen martingale.

\begin{reptheorem}{thm:CI}
	Let~$W$ be a Hermitian operator of the form of \cref{eq:B6} [not necessarily a witness in the sense of \cref{eq:witness_definition_appendix}].
	Suppose that the experiment is modeled by the Model \Cref{assumption:randomness,assumption:measurements,assumption:sequential} in \cref{app:model}.
	Let~$\hat{W}_n$ denote the average witness estimate as defined in \cref{eq:witness_estimate}.
	Fix the significance level~$\alpha\in[0,1]$.
	If~$\alpha < \frac e{2^n}$, define~$\varepsilon = \ds$, otherwise define~$\varepsilon\in[\gamma,\gamma+\ds]$ as the unique solution to
	\begin{equation}
	\alpha = e F^\circ_{n,\half}\Bigl(\frac{n}{2}(1 + \frac{\varepsilon - \gamma}{\ds})\Bigr). \label{eq:vareps}
	\end{equation}
	Here~$\gamma$ and~$F^\circ_{n,\beta}$ are defined in \cref{eq:app gamma,eq:Fc}, respectively (with~$\beta=\half$ here and~$e \approx 2.72$).
	Then,
	\begin{enumerate}[label=(\alph*)]
		\item~$\Pr[\lvert\braket{W}_n  -  \hat{W}_n\rvert \leq \ds] = 1$;
		\item~$\Pr[\braket{W}_n  -  \hat{W}_n \leq \varepsilon] \geq 1 - \alpha$ and  $\Pr[\braket{W}_n  -  \hat{W}_n \geq -\varepsilon] \geq 1-\alpha$;
		\item~$\Pr[\lvert\braket{W}_n  -  \hat{W}_n\rvert \leq \varepsilon] \geq 1 - 2\alpha$.
	\end{enumerate}
	That is, if~$\hat{w}_n$ is the average witness estimate after~$n$ rounds, then~$\Ic(\hat{w}_n) := [\hat{w}_n - \varepsilon, \hat{w}_n + \varepsilon]$ is a~$(1-2\alpha)$ two-sided confidence interval and~$\Jc(\hat{w}_n) := [\hat{w}_n -\ds , \hat{w}_n + \epsilon]$ is a~$(1 - \alpha)$ one-sided confidence interval for the average witness value~$\braket{W}_n$ as defined in \cref{eq:witness_def_app}.
\end{reptheorem}
\begin{proof}
	(a) We show that~$-\ds \leq \braket{W}_n - \hat{W}_n \leq \ds$ holds almost surely by
	\begin{align}
	\braket{W}_n  -  \hat{W}_n
	&=  \frac{1}{n} \sum_{i=1}^{n} \Bigl(\Tr[\rho_i W] - c + S_i \Bigl) \\
	&=  \frac{1}{n} \sum_{i=1}^{n} \Bigl( S_i - \sum_x \sum_{\mathbf{a}} p_x \Tr\Bigl[\bigotimes_{j=1}^m \Pi_{a_j}^{(j),x} \rho_i  \Bigr] s(x, \mathbf{a})  \Bigr)  \\
	&\leq \frac{1}{n} \sum_{i=1}^{n} \Bigl( \smax - \sum_x \sum_{\mathbf{a}} p_x \Tr\Bigl[\bigotimes_{j=1}^m \Pi_{a_j}^{(j),x} \rho_i  \Bigr] \smin  \Bigr) = \ds,
	\end{align}
	since $p_x, \Tr\Bigl[\bigotimes_{j=1}^m \Pi_{a_j}^{(j),x} \rho_i  \Bigr] \geq 0$ and sum to one. Similarly~$\braket{W}_n - \hat{W}_n \geq -\ds$,
	so that
	$
	\Pr[\lvert\braket{W}_n  -  \hat{W}_n\rvert \leq \ds] = 1.
	$

	(b) For~$\alpha < \frac e{2^n}$, both statements in~(b) follow immediately from~(a), since then~$\varepsilon = \ds$.
	So from now on, assume that~$\alpha \in [\frac e{2^n}, 1]$.
	First we show that~$\varepsilon$ is well-defined.
	This follows from \cref{lemma:Fc-props}, which states that for all $y = \frac{\alpha}{e} \in [\frac{1}{2^n}, \frac1e] \subset  [\frac{1}{2^n}, \frac{1}{2}]$, there exists a unique $x = \frac{n}{2}(1+ \frac{\varepsilon - \gamma}{\ds}) \in [\frac{n}{2}, n]$ such that~$F^\circ_{n,\half}(\frac{n}{2}(1+ \frac{\varepsilon - \gamma}{\ds})) = \frac{\alpha}{e}$. Hence for all $\alpha \in [\frac e{2^n}, 1]$ there is a unique $\varepsilon \in [\gamma, \gamma + \ds]$ such that \cref{eq:vareps} holds.

	Now, we construct a suitable martingale sequence~$Z_i$ so that we can apply \cref{lemma:bentkus} to~$Z_i$ and~$-Z_i$ (to get both statements). For~$i=0,\dots,n$ we define
	\begin{equation}
	Z_i := \half\sum_{l=1}^{i} \frac{S_l - \Ebb[S_l | H_l, \rho_l]}{\ds}
	\end{equation}
	This is a martingale sequence with respect to the sequence~$(H_1, \rho_1), \dots, (H_n, \rho_n)$, since
	\begin{align}
	\Ebb[Z_i | H_1,\rho_1,\dots,H_i, \rho_i] &= \frac{1}{2\ds} \sum_{l=1}^{i} \Ebb\Bigl[ S_l - \Ebb[S_l | H_l, \rho_l] \Big| H_1,\rho_1,\dots,H_i, \rho_i\Bigr] \\
  &= \frac{1}{2\ds} \sum_{l=1}^{i} \left( \Ebb\Bigl[ S_l  \Big| H_1,\rho_1,\dots,H_i, \rho_i\Bigr] - \Ebb[S_l | H_l, \rho_l] \right) \label{eq:E13} \\
	&= \frac{1}{2\ds} \left( \Ebb\Bigl[ S_i  \Big| H_1,\rho_1,\dots,H_i, \rho_i\Bigr] - \Ebb[S_i | H_i, \rho_i] \right) \\
  &\qquad + \frac{1}{2\ds} \sum_{l=1}^{i-1} \left( \Ebb\Bigl[ S_l  \Big| H_1,\rho_1,\dots,H_i, \rho_i\Bigr] - \Ebb[S_l | H_l, \rho_l] \right) \\
	&=  \frac{1}{2\ds} \Bigl(\Ebb[S_i |  H_i, \rho_i]  - \Ebb[S_i | H_i, \rho_i] \Bigr)  + \frac{1}{2\ds} \sum_{l=1}^{i-1} \left( S_l - \Ebb[S_l | H_l, \rho_l] \right) \label{eq:D10} \\
	&= 0 + Z_{i-1}.
	\end{align}
	In \cref{eq:E13} we used that $\Ebb[\Ebb[S_l | H_l, \rho_l] | H_1,\rho_1,\dots,H_i, \rho_i] = \Ebb[S_l | H_l, \rho_l]$ for all $l=1,...,i$.
	\Cref{eq:D10} holds since~$\Ebb[S_i | H_1,\rho_1,\dots,H_i,\rho_i] = \Ebb[S_i |  H_i, \rho_i]$ and~$\Ebb[S_l | H_1,\rho_1,\dots,H_i,\rho_i] = S_l$ for all~$l=1,\dots,i-1$ by \cref{assumption:inclusion,assumption:history}.
	Moreover,~$Z_i$ is bounded difference, where the difference is bounded by
	\begin{equation}
	| Z_i - Z_{i-1} | = \half \frac{| S_i  - \Ebb[S_i | H_i, \rho_i] |}{\ds} \leq \half \frac{\ds}{\ds} = \half,
	\end{equation}
	since~$\smin \leq S_i \leq \smax$ for all~$i=1,\dots,n$. Hence \cref{lemma:bentkus} applies with~$\beta = \half$ (to both~$Z_i$ and~$-Z_i$), which yields
	\begin{align}
	\Pr[ Z_n \geq x - \frac{n}{2}] \leq e F^\circ_{n,\half}(x) \qquad \mbox{and} \qquad
	\Pr[ -Z_n \geq x - \frac{n}{2}] \leq e F^\circ_{n,\half}(x). \label{eq:bound_Zn}
	\end{align}
	Next, we invoke \cref{cor:two-sided} to relate~$\braket{W}_n  -  \hat{W}_n$ to~$Z_n$. We find that
	\begin{align}
	\Bigl\lvert \braket{W}_n  -  \hat{W}_n - \frac{2\ds}{n} Z_n \Bigr\rvert &=  \Bigl\lvert \Bigl( \frac{1}{n} \sum_{i=1}^{n}  \Tr[W \rho_i] \Bigr) - \Bigl(c - \frac{1}{n} \sum_{i=1}^n S_i \Bigr) - \Bigl( \frac{1}{n} \sum_{i=1}^{n}  S_i - \Ebb[S_i | H_i, \rho_i] \Bigr)  \Bigr\rvert \\
	&=  \Bigl\lvert \frac{1}{n} \sum_{i=1}^{n} \Tr[W \rho_i] - (c -  \Ebb[S_i | H_i, \rho_i])  \Bigr\rvert \\
	&\leq \frac{1}{n} \sum_{i=1}^{n} \lvert \Tr[W \rho_i] - (c -  \Ebb[S_i | H_i, \rho_i]) \rvert \leq \gamma \label{eq:use_gamma}
	\end{align}
	almost surely. In \cref{eq:use_gamma}, we applied \cref{cor:two-sided} for each~$i=1,\dots,n$. Hence by using \cref{eq:use_gamma,eq:bound_Zn}, we find that
	\begin{align}
	\Pr[ \braket{W}_n  -  \hat{W}_n \geq \varepsilon ] &\leq \Pr[ \gamma + \frac{2\ds}{n} Z_n \geq \varepsilon ] \\
	&= \Pr [Z_n \geq \frac{n(\varepsilon-\gamma)}{2\ds}] \\
	&\leq e F^\circ_{n,\half}\Bigl(\frac{n}{2}(1 + \frac{\varepsilon - \gamma}{\ds})\Bigr) = \alpha, \label{eq:E28}
	\end{align}
	by evaluating \cref{eq:bound_Zn} at the appropriate value of~$x$ and using the implicit definition of~$\varepsilon$ in \cref{eq:vareps}. Similarly,
	\begin{align}
	\Pr[\hat{W}_n - \braket{W}_n \geq \varepsilon]
	&\leq \Pr[ \gamma - \frac{2\ds}{n} Z_n \geq \varepsilon ] \\
	&= \Pr [-Z_n \geq \frac{n(\varepsilon-\gamma)}{2\ds}] \\
	&\leq e F^\circ_{n,\half}\Bigl(\frac{n}{2}(1 + \frac{\varepsilon - \gamma}{\ds})\Bigr) = \alpha. \label{eq:other_side}
	\end{align}
	\Cref{eq:E28,eq:other_side} imply the one-sided intervals
	\begin{equation}
	\Pr[\braket{W}_n - \hat{W}_n \leq \varepsilon] \geq 1-\alpha,
	\qquad \mbox{and} \qquad
	\Pr[\braket{W}_n - \hat{W}_n \geq -\varepsilon] \geq 1-\alpha,
	\end{equation}
	respectively, as claimed.

	(c) Combining \cref{eq:E28,eq:other_side} with the union bound yields
	\begin{equation}
	\Pr[\lvert\braket{W}_n  -  \hat{W}_n\rvert \geq \varepsilon] \leq \Pr[ \braket{W}_n  -  \hat{W}_n \geq \varepsilon ] + \Pr[ \hat{W}_n - \braket{W}_n    \geq \varepsilon ] \leq 2\alpha.
	\end{equation}
	This implies that
	\begin{equation}
	\Pr[ \lvert\braket{W}_n  -  \hat{W}_n\rvert \leq \varepsilon ] \geq 1 - 2\alpha,
	\end{equation}
	as claimed.
\end{proof}

\end{document}